\documentclass[11pt]{article}

\usepackage{amsmath}
\usepackage{amssymb}
\usepackage{amsfonts}
\usepackage{amsthm}
\usepackage[margin=3cm]{geometry}
\usepackage{xcolor}
\usepackage{hyperref}
\usepackage{braket}
\usepackage{enumitem}
\usepackage{graphicx}
\usepackage{epstopdf}
\usepackage{blindtext}
\usepackage{bbm}
\usepackage{mathtools}
\usepackage{cleveref}
\usepackage{tabularx}
\usepackage{bm}
\usepackage{multicol}

\usepackage{natbib}
\allowdisplaybreaks

\hypersetup{colorlinks=true,linkcolor=blue,urlcolor=blue,citecolor=blue,anchorcolor=blue}

\newcommand{\rr}{\mathbb{R}}

\newcommand{\cc}{\mathbb{C}}
\newcommand{\nn}{\mathbb{N}}
\newcommand{\rank}{\text{rank}}
\newcommand{\hilbert}{\textbf{H}}
\newcommand{\hilb}{\textbf{H}}
\newcommand{\bounded}{\mathcal{B}}
\newcommand{\linear}{\mathcal{L}}
\newcommand{\linearhs}{\mathcal{L}_{HS}}
\newcommand{\linearhsh}{\mathcal{L}_{HS}(\textbf{H})}
\newcommand{\linearhshh}{\mathcal{L}_{HS} (\textbf{H}_1, \textbf{H}_2)}

\newcommand{\norm}[1]{\left\lVert #1 \right\rVert}
\newcommand{\normhs}[1]{\left\lVert #1 \right\rVert_{HS}}
\newcommand{\limk}{\xrightarrow{k\to \infty}}
\newcommand{\limn}{\xrightarrow{n\to \infty}}

\newcommand{\bddhh}{\bounded(\hilbert_1, \hilbert_2)}
\newcommand{\bddh}{\bounded(\hilbert)}

\newcommand{\hilone}{\hilbert_1}
\newcommand{\hiltwo}{\hilbert_2}
\newcommand{\hili}{\hilbert_i}
\newcommand{\hilN}{\hilbert_N}
\newcommand{\tpN}{\hilbert_1\otimes\cdots\otimes\hilbert_N}
\newcommand{\tptwo}{\hilbert_1\otimes\hilbert_2}
\newcommand{\otimesdots}{\otimes\cdots\otimes}

\newcommand{\lr}[1]{\left( #1 \right)}

\newcommand{\indic}{\mathbbm{1}}
\newcommand{\superone}{^{(1)}}

\newcommand{\sumk}{\sum_{k=1}^\infty}

\newcommand{\ketpsi}{\ket{\psi}}
\newcommand{\ketphi}{\ket{\phi}}

\newcommand{\omegatilde}{\Tilde{\omega}}

\newcommand{\ketpsione}{\ket{\psi^{(1)}_{a_m}}}
\newcommand{\ketpsitwo}{\ket{\psi^{(2)}_{a_m}}}

\newcommand{\hermite}[1]{\mathcal{H}_{#1}}

\theoremstyle{plain}
\newtheorem{theom}{Theorem}[section]

\newtheorem{prop}[theom]{Proposition}

\newtheorem{lem}[theom]{Lemma}

\theoremstyle{definition}

\newtheorem{defin}[theom]{Definition}

\newtheorem{remark}[theom]{Remark}

\numberwithin{equation}{section}

\usepackage[utf8]{inputenc} 
\usepackage{titling}

\title{The Canonical Forms of Matrix Product States in Infinite-Dimensional Hilbert Spaces}
\author{ {
Niilo Heikkinen
} \\
\small{Department of Mathematics and Statistics, University of Helsinki, niilo.heikkinen@helsinki.fi} 
}

\begin{document}
\maketitle
\vspace{-6mm}
\section*{Abstract} 
In this work, we prove that any element in the tensor product of separable infinite-dimensional Hilbert spaces can be expressed as a matrix product state (MPS) of possibly infinite bond dimension.
The proof is based on the singular value decomposition of compact operators and the connection between tensor products and Hilbert-Schmidt operators via the Schmidt decomposition in infinite-dimensional separable Hilbert spaces. The construction of infinite-dimensional MPS (idMPS) is analogous to the well-known finite-dimensional construction in terms of singular value decompositions of matrices. The infinite matrices in idMPS give rise to operators acting on (possibly infinite-dimensional) auxiliary Hilbert spaces. As an example we explicitly construct an MPS representation for certain eigenstates of a chain of three coupled harmonic oscillators. 
\vspace{-2mm}
\tableofcontents

\section{Introduction}
In the last couple of decades, tensor networks have played a central role in the study of many-body quantum systems by providing efficient parametrizations of state vectors in a tensor product space by expressing higher rank tensors as "networks" of tensors of lower rank (for reviews see. eg. \cite{schollwock2011density}, \cite{bridgeman2017hand}, \cite{orus2014practical} and \cite{cirac2021matrix} or alternatively \cite{hackbusch2012tensor} for a more mathematical presentation). Matrix product states (MPS) are simple, but nontrivial tensor networks, and they are applicable in the study of realistic one-dimensional quantum systems. In particular, they are central in the celebrated density matrix renormalization group (DMRG) algorithm \cite{schollwock2011density} and they have been especially useful in studies of ground states of one-dimensional gapped Hamiltonians, which satisfy an area law for entanglement entropy (\cite{hastings2007area}, \cite{verstraete2004matrix}, \cite{vidal2004efficient}). Furthermore, MPS have also allowed for generalizations such as the so-called continuous MPS (cMPS), which arises as the continuous limit of MPS and is used in the study of one-dimensional quantum field theories (\cite{verstraete2010continuous}, \cite{haegeman2013calculus}, \cite{tilloy2019continuous}). Additionally, MPS have found several applications beyond these (\cite{ritter2024quantics}, \cite{fernandez2025learning}). In the numerical and mathematical communities, MPS are known as tensor trains (TT) \cite{oseledets2011tensor}. Furthermore, MPS/TT can be understood as a special case of the hierarchical tucker (HT) form of a tensor (\cite{hackbusch2012tensor}, Sec. 12.2).

 An MPS is a representation of a vector in a tensor product space where the expansion coefficients with respect to a basis are given as a certain product of matrices. Depending on the chosen boundary conditions, either the dimensions of the matrices are such that their product results in a scalar, or we take the trace of their product. The maximum dimension of the associated matrices is referred to as the bond dimension of the MPS.

It is a well known result that an arbitrary vector in an $N$-fold tensor product of finite-dimensional Hilbert spaces can be written as an MPS, and the decomposition can be obtained iteratively by repeated singular value or Schmidt decompositions (\cite{schollwock2011density}, \cite{vidal2003efficient}). In this paper we perform a straightforward generalization of this procedure for a vector in an infinite-dimensional separable Hilbert space. We will use 
a method that is analogous to the finite-dimensional construction and allows us to obtain an exact MPS representation of possibly infinite bond dimension. 

Related to tensor networks with infinite bond dimension, a recent study explored tensor renormalization group maps in the context of infinite-dimensional Hilbert spaces \cite{rychkov2022tensor}. Additionally, MPS-based methods have been previously used in the study of systems with a finite number of constituents, but each with continuous degrees of freedom, using finite-dimensional MPS as approximations for these systems \cite{iblisdir2007mpscont}. The result in this paper proves that these approximations converge to the original state vector in the norm of the tensor product Hilbert space. It should be noted that this discrete approximation has to be performed by truncating both the physical and the auxiliary (virtual) degrees of freedom separately, and that the matrix elements in the MPS do not always decay monotonically, as seen in Section 4 of this paper.

The MPS decomposition is not unique, and MPS possess a gauge degree of freedom. This gauge freedom allows for us to always write the MPS in any of the so-called canonical forms, which make certain computations of e.g. expected values and matrix elements straightforward \cite{schollwock2011density}. The canonical forms have natural generalizations in the infinite-dimensional context. Additionally, the infinite matrices in idMPS give rise to operators acting on auxiliary Hilbert spaces such that their product results in a scalar.

The paper is organized as follows. First, in section 2 we recall elementary results and state the relevant definitions. In Section 3 we state and prove our main result, mainly that an arbitrary vector in a tensor product of separable Hilbert spaces can be written as an idMPS in any of the canonical forms. In Sections 3.1--3.2 we construct each of the canonical forms by applying Schmidt decompositions in the infinite-dimensional context. In Section 3.3 we briefly discuss the operators that arise from idMPS. In Section 4, we construct an idMPS representation for certain eigenstates of a chain of three quantum harmonic oscillators. Finally, in Section 5 we provide conclusions and an outlook on possible future directions.

\section{Mathematical Background}
In this section we establish notation and cite known mathematical results and definitions that are used in proving the main results of this paper. 

We denote by $\hilb$ or $\hilb_n,$ where $ n\in\nn,$ a separable Hilbert space over the complex field, and $\hilb^*$ denotes the dual of $\hilb.$ We denote by $\bddhh$ and $\bddh$ the sets of bounded operators from $\hilone$ to $\hiltwo$ and bounded operators from $\hilb$ to $\hilb,$ respectively. By $\linearhshh$ and $\linearhsh$ we denote the sets of Hilbert-Schmidt operators from $\hilone$ to $\hiltwo$ and from $\hilb$ to $\hilb,$ respectively. We use Dirac's braket notation, and inner products are assumed to be linear in the second argument.
The tensor product of two kets $\ketpsi$ and $\ketphi$ is denoted by any of the following expressions: $\ketpsi\otimes\ketphi = \ketpsi\ketphi = \ket{\psi,\phi}.$

Our construction of idMPS relies on the existence of the Schmidt decomposition in general separable Hilbert spaces, proof of which can be found in the Appendix.

\begin{theom} [Schmidt Decomposition] \label{SD-Prop}
For any $\ketpsi\in \hilbert_1 \otimes \hilbert_2,$ there exist orthonormal sets $\{ \ket{e_k} \}_{k=1}^{N} \subseteq \hilbert_1$ and $\{ \ket{f_k} \}_{k=1}^{N} \subseteq \hilbert_2$  where $N\in\nn_0\cup\{\infty\},$ as well as nonnegative real numbers $\{ \lambda_k \}_{k=1}^{N},$ with $\lambda_k \limk 0$ (if $N$ is infinite), such that
\begin{equation} \label{SD}
\ketpsi = \sum_{k=1}^{N} \lambda_k \ket{e_k} \otimes \ket{f_k},
\end{equation}
with convergence in the norm of $\tptwo$. The numbers $\lambda_k$ are called \emph{Schmidt coefficients}, the vectors $\ket{e_k}$ and $\ket{f_k}$ the \emph{left and right Schmidt vectors}, respectively, and the expression (\ref{SD}) a \emph{Schmidt decomposition} (SD) of $\ketpsi.$
\end{theom}
\begin{proof} See Appendix \ref{Appendix A}.
\end{proof}

Let us recall the definition of (finite-dimensional) matrix product states.
\begin{defin} [Matrix Product State]
Let $\hilone,\dots,\hilN$ be finite dimensional Hilbert spaces of dimensions $\dim(\hili)=d_i$ with orthonormal bases $\{ \ket{k_i} \}_{k_i=0}^{d_i-1}$ for each. 
A vector 
\begin{equation}
\ket{\psi} = \sum_{k_1=0}^{d_1-1}\cdots\sum_{k_N=0}^{d_N-1} c_{k_1,...,k_N} \ket{k_1,\dots,k_N} \in \hilone\otimes\cdots\otimes\hilN
\end{equation}
is called a \emph{matrix product state with open boundary conditions}, if the coefficients are written as
 \begin{equation} \label{mps-indices-c}
     c_{k_1,...,k_N} = M^{(k_1)} \Lambda^{(1)} M^{(k_2)} \Lambda^{(2)} \cdots \Lambda^{(N-1)} M^{(k_N)},
 \end{equation}
in terms of $dN$ complex matrices 
$\{ M^{(k_i)} \: | \: i=1,...,N, \: k_i=1,\dots,d_i \}$ 
and $N-1$ real diagonal
matrices 
$\{ \Lambda^{(i)} \: | \: i=1,...,N-1\}.$
 If the diagonal matrices are equal to the identity, they are not written explicitly and then simply
  \begin{equation}
     c_{k_1,...,k_N} = M^{(k_1)} M^{(k_2)} \cdots M^{(k_N)}.
 \end{equation}
A vector $\ketpsi$ is a \emph{matrix product state with periodic boundary conditions} if the coefficients are written in the form 
 \begin{equation} \label{mps-indices-c3}
     c_{k_1,...,k_N} = \text{tr}\big(M^{(k_1)} M^{(k_2)} \cdots M^{(k_N)}\big),
 \end{equation}
 where $\{ M^{(k_i)} \: | \: i=1,\dots,N, \: k_i=1,\dots,d_i \}$ are complex square matrices.
\end{defin}

\begin{remark} [Notation]
\begin{enumerate}
    \item In the above definition we used the notation that is standard in the literature, and did not write the site indices of the matrices $M^{(k_n)}$ explicitly. It would be more precise to write $M^{(n,k_n)}$ instead of $M^{(k_n)},$ as then the index $n$ would reveal which set of matrices is considered, and the index $k_n$ which matrix in particular. However, for our purposes the extra site index is somewhat redundant, and as long as we are not writing any explicit numerical values, we can identify the set of matrices from the index $k_n$. To make notation less cumbersome, we will stick to this convention.
    \item In (\ref{mps-indices-c})--(\ref{mps-indices-c3}) the indices appear as lower indices on the left hand side and as upper indices on the right hand side. This is nothing but a notational convention, and the positioning of indices as upper or lower ones does not make a difference in the context of this paper. It is somewhat common in the MPS literature to denote the physical indices $k_n$ as upper indices (see e.g. \cite{schollwock2011density}, \cite{cirac2021matrix}), reserving the lower indices for the virtual (bond) indices of the matrices, since this makes the matrix/operator product interpretation very clear.  In this paper we follow this convention. Additionally, the upper index of each of the diagonal matrices $\Lambda^{(n)}$ is used to differentiate between different matrices. Again, any lower indices denote the rows and columns of the matrix in question.
\end{enumerate}
\end{remark}

It is common to write tensor networks using \emph{tensor diagram notation,} in which tensors are represented by graphs where the nodes represent the tensors and the edges (sometimes called "legs") the indices, and connected indices are contracted (for more details on tensor diagram notation see e.g. \cite{biamonte2017tensor}).
In this notation, a general element $\ketpsi\in\tpN$ is represented by a single node with $N$ edges attached to it:
\begin{equation} 
    \includegraphics[]{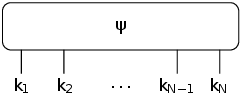}
\end{equation}
and an MPS with open boundary conditions is written as a graph of the form
\begin{equation} \label{canonical-mps-diagram1}
    \includegraphics[]{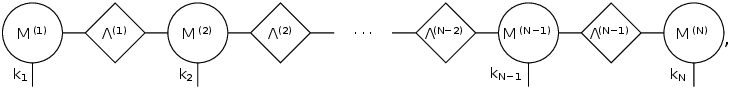}
\end{equation}
or with the diagonal matrices $\Lambda^{(n)}$ set to unity:
\begin{equation}
    \includegraphics[]{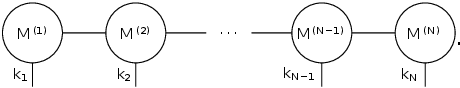}
\end{equation}
The corresponding MPS with periodic boundary conditions is written as the graph
\begin{equation}
    \includegraphics[]{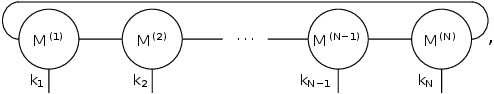}
\end{equation}
where the line connecting $M^{(1)}$ and $M^{(N)}$ corresponds to the trace operation. When writing MPS as diagrams, the notation is far less cluttered, and in the diagrammatic representation we write explicitly $M^{(n,k_n)}=M^{(k_n)}$ to emphasize that the matrices are in general different.

\begin{remark}
    In a matrix product state with open boundary conditions, the matrices $M^{(k_1)}$ and $M^{(k_N)}$ are actually row and column vectors, respectively, ensuring that the matrix product results in a scalar.
\end{remark}

In this paper we focus on MPS with open boundary conditions.
We conclude this section by defining the canonical forms of MPS in infinite-dimensional Hilbert spaces.

\begin{defin} [Infinite-Dimensional MPS] \label{idMPS}
Let $\hilone,\dots,\hilN$ be separable Hilbert spaces with orthonormal bases $\{ \ket{k_n} \}_{k_n=0}^{d_n-1}$ for each, and assume that at least one of $\hilb_n$ is infinite-dimensional (that is, $d_n=\infty$ for some $n=1,\dots,N$). A vector 
\begin{equation}
\ket{\psi} = \sum_{k_1=0}^{d_1-1}\cdots\sum_{k_N=0}^{d_N-1} c_{k_1,...,k_N} \ket{k_1,\dots,k_N} \in \hilone\otimes\cdots\otimes\hilN
\end{equation}
is called an \emph{infinite-dimensional matrix product state} (idMPS), if the coefficients are written in the form 
\begin{equation}
    c_{k_1,...,k_N} = \sum_{a_1=0}^\infty\cdots\sum_{a_N=0}^\infty M^{(k_1)}_{a_1} \Lambda^{(1)}_{a_1,a_1} M^{(k_2)}_{a_1,a_2} \Lambda^{(2)}_{a_2,a_2}\cdots M^{(k_{N-1})}_{a_{N-2},a_{N-1}} \Lambda^{(N-1)}_{a_{N-1},a_{N-1}} M^{(k_N)}_{a_{N-1}},
    \vspace{2mm}
\end{equation}
where for any values of $a_n$ we have $M^{(k_n)}_{a_{n-1},a_n}, \in \cc$ for $n\in\{2,\dots,N-1\}$ and $M^{(k_1)}_{a_1}\in\cc$ as well as  $M^{(k_N)}_{a_{N-1}}\in\cc.$ Also, we require $\Lambda^{(n)}_{a_n,a_n} \in \rr$ for any $n\in\{ 1,\dots,N-1 \}.$ If $ \Lambda^{(n)}_{a_n,a_n}=1$ for all values of $a_n,$ we do not write them explicitly, and write simply
\begin{equation}
    c_{k_1,...,k_N} = \sum_{a_1=0}^\infty\cdots\sum_{a_N=0}^\infty M^{(k_1)}_{a_1} M^{(k_2)}_{a_1,a_2} \cdots M^{(k_{N-1})}_{a_{N-2},a_{N-1}} M^{(k_N)}_{a_{N-1}}.
\end{equation}
\end{defin}

Let us recall the four different canonical forms of MPS as in \cite{schollwock2011density}. We state the definitions explicitly in terms of components of the matrices instead of referring to adjoints of the associated operators.
\begin{defin} [Left-Canonical MPS]\label{LC-MPS}
	A matrix product state
	\begin{equation}
		\ket{\psi} =
		\sum_{k_1=0}^{d_1-1} \cdots \sum_{k_N=0}^{d_N-1} 
		\sum_{a_1=0}^\infty\cdots\sum_{a_{N-1}=0}^\infty A^{(k_1)}_{a_1} A^{(k_2)}_{a_1,a_2} \cdots A^{(k_{N-1})}_{a_{N-2},a_{N-1}} A^{(k_N)}_{a_{N-1}}
  \ket{k_1, k_2, \dots, k_N}
	\end{equation}
	is called a \emph{left-canonical MPS} if all of the matrices are \emph{left-normalized,} i.e. if for every $n\in\{ 1, \dots, N \}$
	\begin{equation}
        \sum_{k_n}\sum_{a_{n-1}}  \left( A^{(k_n)} \right)^*_{a_{n-1},a_n} A^{(k_n)}_{a_{n-1},b_n} = \delta_{a_n b_n},
	\end{equation}
	where we set dummy indices $a_0=a_N=1$, $\delta_{a_n b_n}$ is the Kronecker delta and the asterisk denotes complex conjugation. We will use the letter $A$ to denote left normalized matrices.
\end{defin}

In tensor diagram notation, we denote a left-canonical MPS by a directed graph of the form
\begin{equation}
    \includegraphics[]{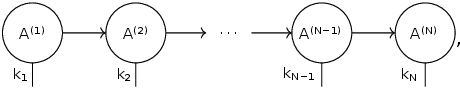}
\end{equation}
where the arrows denote the direction of normalization.

\begin{defin} [Right-Canonical MPS]\label{RC-MPS}
	A matrix product state
	\begin{equation}
		\ket{\psi} =
		\sum_{k_1=0}^{d_1-1} \cdots \sum_{k_N=0}^{d_N-1} 
		\sum_{a_1=0}^\infty\cdots\sum_{a_{N-1}=0}^\infty B^{(k_1)}_{a_1} B^{(k_2)}_{a_1,a_2} \cdots B^{(k_{N-1})}_{a_{N-2},a_{N-1}} B^{(k_N)}_{a_{N-1}}
  \ket{k_1, k_2, \dots, k_N}.
	\end{equation}
	is called a \emph{right-canonical MPS} if all of the matrices are \emph{right-normalized,} i.e. if for every $n\in\{ 1, \dots, N \}$
	\begin{equation}
        \sum_{k_n}\sum_{a_{n}} B^{(k_n)}_{a_{n-1},a_n} \left( B^{(k_n)}_{b_{n-1},a_n} \right)^*  = \delta_{a_{n-1} b_{n-1}},
	\end{equation}
	where we set dummy indices $a_0=a_N=1$, $\delta_{a_{n-1} b_{n-1}}$ is the Kronecker delta and the asterisk denotes complex conjugation. We will use the letter $B$ to denote right normalized matrices.
\end{defin}

A right-canonical MPS is denoted by a directed graph of the form
\begin{equation}
    \includegraphics[]{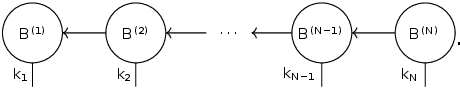}
\end{equation}

\begin{defin}[Mixed-Canonical MPS]\label{MC-MPS}
    Fix $n\in\{ 1, \dots, N-1 \}.$ A matrix product state of the form
    \begin{equation}
	\ket{\psi} \hspace{-0.5mm} = \hspace{-0.5mm}
	\sum_{k_1=0}^{d_1-1} \hspace{-0.5mm} \cdots \hspace{-0.5mm} \sum_{k_N=0}^{d_N-1} \sum_{a_1=0}^\infty \hspace{-0.5mm} \cdots\hspace{-0.5mm}\sum_{a_{N-1}=0}^\infty 
 A^{(k_1)}_{a_1} 
 \cdots A^{(k_n)}_{a_{n-1},a_{n}} D_{a_n,a_n}
 B^{(k_{n+1})}_{a_{n},a_{n+1}} \cdots 
 B^{(k_N)}_{a_{N-1}}
  \ket{k_1, k_2, \dots, k_N}
	\end{equation}
    is called a \emph{mixed-canonical MPS} if the matrices $\{A^{(k_1)}, \dots, A^{(k_n)}\}$ are left-normalized, the matrices $\{ B^{(k_{n+1})}, \dots, B^{(k_N)}\}$ are right-normalized and $D_{a_n,a_n}\geq 0$ (in particular $D_{a_n,a_n}\in\rr$) for every $a_n\in\nn_0.$
Defining for $n=0$ and $n=N$ the "matrix" $D$ as the scalar $D_{a_0,a_0}=1=D_{a_N,a_N}$ yields the left- and right-canonical forms as the $n=0$ and $n=N$ special cases of the mixed-canonical form, respectively.
\end{defin}

A mixed-canonical MPS is denoted by a directed graph of the form
\begin{equation}
    \includegraphics[]{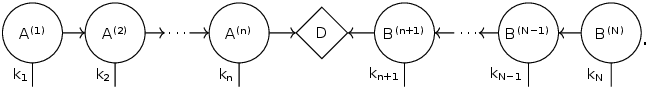}
\end{equation}

\begin{defin}[Canonical MPS]\label{C-MPS} 
Consider a matrix product state of the form
\begin{equation}
\ket{\psi} = \sum_{k_1=0}^{d_1-1}\cdots\sum_{k_N=0}^{d_N-1} c_{k_1,\dots,k_N} \ket{k_1,\dots,k_N} \in \tpN,
\end{equation}
where
\begin{equation}
    c_{k_1,...,k_N} = \sum_{a_1=0}^\infty\cdots\sum_{a_{N-1}=0}^\infty \Gamma^{(k_1)}_{a_1} \Lambda^{(1)}_{a_1,a_1} \Gamma^{(k_2)}_{a_1,a_2} \Lambda^{(2)}_{a_2,a_2}\cdots \Gamma^{(k_{N-1})}_{a_{N-2},a_{N-1}} \Lambda^{(N-1)}_{a_{N-1},a_{N-1}} \Gamma^{(k_N)}_{a_{N-1}}.
\end{equation}
The state $\ket{\psi}$ is called a \emph{canonical MPS,} if for any $n\in \{ 1,\dots,N-1 \}$ the expression
 \begin{equation}
	\sum_{a_n} \lambda^{(n)}_{a_n} \ket{\phi_{a_n}^{(1,...,n)} } \otimes \ket{\phi_{a_n}^{(n+1,...,N)} } ,
  \vspace{-5mm}
	\end{equation}

	where 
	\begin{align}
  \lambda^{(n)}_{a_n} &= \Lambda^{(n)}_{a_n,a_n}, \nonumber
  \\
\ket{\phi_{a_n}^{(1,...,n)} } 
&= \sum_{k_1} \cdots \sum_{k_n} \sum_{a_1}\cdots\sum_{a_{n-1}}
   \Gamma^{(k_1)}_{a_1} \Lambda^{(1)}_{a_1,a_1} 
  \cdots 
   \Lambda^{(n-1)}_{a_{n-1},a_{n-1}} \Gamma^{(k_n)}_{a_{n-1},a_n} \ket{k_1,\dots,k_n}, 
  \nonumber \\
\ket{\phi_{a_n}^{(n+1,...,N)} } \hspace{-0.5mm}
&= \hspace{-1mm} \sum_{k_{n+1}} \hspace{-1.0mm}\cdots\hspace{-1.0mm} \sum_{k_N} \sum_{a_{n+1}}\hspace{-1.0mm}\cdots\hspace{-1mm}\sum_{a_{N-1}}
   \Gamma^{(k_{n+1})}_{a_n,a_{n+1}} \Lambda^{(n+1)}_{a_{n+1},a_{n+1}} \cdots \Lambda^{(N-1)}_{a_{N-1},a_{N-1}} \Gamma^{(k_N)}_{a_{N-1}}\ket{k_{n+1},\dots,k_N},\nonumber
	\end{align}
	is a Schmidt decomposition of $\ket{\psi}$ with respect to the partition $(\hilone\otimesdots\hilb_n)\otimes\\(\hilb_{n+1}\otimesdots\hilN).$\footnote{By this we mean that we identify $\hilone\otimes\cdots\otimes\hilb_n\otimes\cdots\otimes\hilN$ with $(\hilone\otimes\cdots\hilb_n)\otimes(\hilb_{n+1}\cdots\otimes\hilN)$ and take the Schmidt decomposition of $\ket{\psi} \in (\hilone\otimes\cdots\hilb_n)\otimes(\hilb_{n+1}\cdots\otimes\hilN).$ In this way we can apply the Schmidt decomposition, which applies to two-fold tensor products, to N-fold tensor products. The validity of this operation is justified by the associativity of the tensor product of Hilbert spaces.}
\end{defin}
 In tensor diagram notation, a canonical MPS is written in the form (\ref{canonical-mps-diagram1}) with $M=\Gamma.$

\section{The Main Result}
Any state vector in a tensor product of separable Hilbert spaces can be written as an MPS in any of the canonical forms. This is the content of the following Theorem. The proof is the content of Sections 3.1 and 3.2, where we construct the canonical forms of an arbitrary state vector using a method analogous to the finite-dimensional case.

\begin{theom}
\label{idMPS existence}
Let $\hilone,\dots,\hilN$ be separable Hilbert spaces. Any
\begin{equation} \label{ketpsigeneral}
\ket{\psi} = \sum_{k_1=0}^\infty \cdots \sum_{k_N=0}^\infty c_{k_1,\dots,k_N} \ket{k_1,\dots,k_N} \in \hilone\otimes\cdots\otimes\hilN    
\end{equation}
can be written as a matrix product state in any of the canonical forms given in Definitions \ref{LC-MPS}, \ref{RC-MPS}, \ref{MC-MPS} and \ref{C-MPS}.
\end{theom}

\subsection{Construction of Left-, Right- and Mixed-Canonical idMPS}
\begin{proof} [Proof of Theorem \ref{idMPS existence} for the left-, right- and mixed-canonical forms]
Consider $\ketpsi$ given in \\(\ref{ketpsigeneral}).
Let us fix $m\in\{0,1,\dots,N\}$ and construct a  corresponding mixed-canonical MPS representation of $\ket{\psi}.$ 
As the very first thing we Schmidt decompose $\ketpsi$ with respect to the partition $(\hilone\otimesdots\hilb_m)\otimes(\hilb_{m+1}\otimesdots\hilN)$ as
\begin{equation} \label{mc-sd-m}
    \ketpsi = \sum_{a_m} \lambda_{a_m} \ketpsione\ketpsitwo,
\end{equation}
where $\ketpsione\in\hilone\otimesdots\hilb_m$ and $\ketpsitwo \in \hilb_{m+1}\otimesdots\hilN$ are orthonormal sets and $\lambda_{a_m}\geq 0.$ The series (\ref{mc-sd-m}) converges to $\ketpsi$ in the tensor product Hilbert space, and thus the coefficients form an $\ell^2$-sequence, i.e. $\sum_{a_m}\lambda_{a_m}^2 < \infty.$ 

If $m=0$ or $m=N,$ then (\ref{mc-sd-m}) is a trivial Schmidt decomposition,
and there is only one nonzero term in the series, which is equal to $\ketpsi.$ If $m=0,$ then $\ketpsi=\ketpsitwo$ and we skip part $ii.)$ of the proof below. Similarly, if $m=N,$ then $\ketpsi=\ketpsione$ and we skip part $i.)$ of the proof.

The proof will proceed in four stages. First we construct a right-canonical idMPS representation of $\ketpsitwo$ using a process called \emph{right leaf stripping,} and then we construct a left-canonical idMPS representation of $\ketpsione$ using \emph{left leaf stripping.} After this, we combine the expressions into a mixed canonical idMPS representation of $\ketpsi$ and conclude the proof by verifying the normalization conditions.
\\\\
$i.)$ \emph{Right leaf stripping.}
\\
We construct right-normalized, possibly infinite matrices $B^{(k_{m+1})},\dots,B^{(k_N)}$ by iteratively applying Schmidt decompositions to the vector $\ketpsitwo\in\hilb_{m+1}\otimesdots\hilN$ with respect to the partitions $\lr{\hilb_{m+1}\otimesdots\hilb_{N-1}}\otimes\lr{\hilN},\dots,\lr{\hilb_{m+1}}\otimes\lr{\hilb_{m+2}\otimesdots\hilN}$ according to the steps outlined below.
\\
$\bm{1.)}$\\ 
    We Schmidt decompose $\ketpsitwo$ with respect to the partition $(\hilb_{m+1} \otimes\cdots\otimes \hilb_{N-1})\otimes(\hilb_N)$ as
    \begin{equation} \label{RC-SD-1-}
        \ketpsitwo = \sum_{a_{N-1}} \lambda_{a_m,a_{N-1}} \ket{x_{a_{N-1}}^{(m+1,\dots,N-1)}}\ket{y_{a_{N-1}}^{(N)}}.
    \end{equation}
    The series (\ref{RC-SD-1-}) converges in the tensor product Hilbert space and the dependence on $a_m$ is encoded in the Schmidt coefficients. If $m=0,$ then we can set $a_m=1$ or omit it altogether.

    We expand the right Schmidt vectors $\ket{y_{a_{N-1}}^{(N)}}\in\hilN$ 
    in the basis $\{ \ket{k_N} \}$ as 
    \begin{equation}\label{rc_y1-}
        \ket{y_{a_{N-1}}^{(N)}} = \sum_{k_N} B_{a_{N-1}}^{(k_N)} \ket{k_N}
    \end{equation}
    and apply multilinearity and separate continuity of the tensor product to obtain
    \begin{align}
        \ketpsitwo &= \sum_{a_{N-1}} \sum_{k_N} \lambda_{a_m,a_{N-1}} 
        B_{a_{N-1}}^{(k_N)} \ket{x^{(m+1,\dots,N-1)}_{a_{N-1}}}\ket{k_N} \\
        &= \sum_{k_N} \sum_{a_{N-1}} \lambda_{a_m,a_{N-1}} B_{a_{N-1}}^{(k_N)} \ket{x^{(m+1,\dots,N-1)}_{a_{N-1}}}\ket{k_N}.
    \end{align}
    Note that the $k_i$-sums are associated with a unitary change of basis and can therefore be exchanged as above without affecting convergence. From now on we will implicitly use this fact whenever rearranging summations.
    \\
$\bm{2.)}$\\
We take a Schmidt decomposition of $\ketpsitwo$ with respect to the partition $(\hilb_{m+1} \otimes\cdots\otimes \hilb_{N-2})\otimes(\hilb_{N-1}\otimes\hilb_N)$, given by the formula
    \begin{equation} \label{RC-SD-2-}
        \ketpsitwo = \sum_{a_{N-2}} \lambda_{a_m,a_{N-2}} \ket{x_{a_{N-2}}^{(m+1,\dots,N-2)}}\ket{y_{a_{N-2}}^{(N-1,N)}},
    \end{equation}
    where $\{\ket{x_{a_{N-2}}^{(m+1,\dots,N-2)}}\}\subseteq \hilb_{m+1}\otimes\cdots\otimes\hilb_{N-2},$ $\{\ket{y_{a_{N-2}}^{(N-1,N)}}\} \subseteq \hilb_{N-1}\otimes\hilN$
    and the series (\ref{RC-SD-2-}) converges to $\ketpsitwo$ similarly as the series (\ref{RC-SD-1-}).
    The orthonormal set $\{\ket{y_{a_{N-1}}^{(N)}}\} \subseteq \hilN$ obtained in step 1 can be extended to an orthonormal basis of $\hilN,$ which we denote simply by $\{\ket{y_{a_{N-1}}^{(N)}}\}.$ Hence we can expand $\ket{y_{a_{N-2}}^{(N-1,N)}}$
    in the basis $\{ \ket{k_{N-1}}\otimes\ket{y_{a_{N-1}}^{(N)}} \},$ and we denote the coefficients by $B_{a_{N-2},a_{N-1}}^{(k_{N-1})}.$
    Additionally, the basis vectors $\ket{y_{a_{N-1}}^{(N)}}$ can be written in the form (\ref{rc_y1-}). We obtain the following expression:
    \begin{align}
        \ketpsitwo 
        &= \sum_{k_{N-1}}\sum_{k_{N}} \sum_{a_{N-2}}\sum_{a_{N-1}} \lambda_{a_m,a_{N-2}}   B_{a_{N-2},a_{N-1}}^{(k_{N-1})} B_{a_{N-1}}^{(k_N)} \ket{x_{a_{N-2}}^{(m+1,\dots,N-2)}} \ket{k_{N-1},k_N}.
    \end{align}
In particular, we obtained the following expression for the right Schmidt vectors:
    \begin{equation} \label{rc_y2-}
        \ket{y_{a_{N-2}}^{(N-1,N)}} = \sum_{k_{N-1}}\sum_{k_N} \sum_{a_{N-1}} B_{a_{N-2},a_{N-1}}^{(k_{N-1})} B_{a_{N-1}}^{(k_N)} \ket{k_{N-1}k_N}.
    \end{equation}
$\bm{n=3,...,N-m-1.)}$\\
We proceed as in step 2, by taking a Schmidt decomposition of $\ketpsitwo$ with respect to the partition $(\hilb_{m+1} \otimes\cdots\otimes \hilb_{N-n})\otimes(\hilb_{N-n+1}\otimes\cdots\otimes\hilb_N)$ and expanding the right Schmidt vectors $\ket{y_{a_{N-n}}^{(N-n+1,\dots,N)}}$ 
in the basis $\{\ket{k_{N-n+1}}\otimes\ket{y_{a_{N-n+1}}^{(N-n+2,\dots,N)} }\} \subseteq \hilb_{N-n+1}\otimes\cdots\otimes\hilN$
to obtain an expression of the form
\begin{equation}
        \ketpsitwo
        = \sum_{a_{N-n}} 
        \sum_{k_{N-n+1}} \sum_{a_{N-n+1}} \hspace{-1mm} \lambda_{a_m,a_{N-n}} B_{a_{N-n},a_{N-n+1}}^{(k_{N-n+1})}
        \ket{x_{a_{N-n}}^{(m+1,...,N-n)}} \ket{k_{N-n+1}}
        \ket{y_{a_{N-n+1}}^{(N-n+2,\dots,N)}},
\end{equation}
where $B_{a_{N-n},a_{N-n+1}}^{(k_{N-n+1})}$ are the associated expansion coefficients of $\ket{y_{a_{N-n}}^{(N-n+1,\dots,N)}}.$

Writing $\ket{y_{a_{N-n+1}}^{(N-n+2,\dots,N)}}$ in the basis $\{\ket{k_{N-n+2},\dots,k_N}\}$ in the form obtained in the previous step (in the form (\ref{rc_y2-})
\footnote{For general $n$, \[ \ket{y_{a_{N-n}}^{(N-n+1,\dots,N)}} = \sum_{k_{N-n+1}}\cdots\sum_{k_N} \sum_{a_{N-1}} \cdots\sum_{a_{N-n+1}}  B_{a_{N-n},a_{N-n+1}}^{(k_{N-n+1})} B_{a_{N-n+1},a_{N-n+2}}^{(k_{N-n})} \cdots B_{a_{N-1}}^{(k_N)} \ket{k_{N-n+1},\dots,k_N}\] })
as well as reordering the sums yields
\begin{multline}
       \ketpsitwo =  \sum_{k_{N-n+1}} \cdots \sum_{k_{N}} \sum_{a_{N-n}} 
       \cdots \sum_{a_{N-1}}
       \lambda_{a_m,a_{N-n}} B_{a_{N-n},a_{N-n+1}}^{(k_{N-n+1})} \cdots \\B_{a_{N-2},a_{N-1}}^{(k_{N-1})} B_{a_{N-1}}^{(k_N)} 
       \ket{x_{a_{N-n}}^{(m+1,...,N-n)}}\ket{k_{N-n+1},\dots,k_N}.
\end{multline}
In particular, we obtained the following expression for the right Schmidt vectors:
    \begin{equation} \label{rc_yn}
        \ket{y_{a_{N-n}}^{(N-n+1,\dots,N)}} = \sum_{k_{N-n+1}} \cdots\sum_{k_N} \sum_{a_{N-n+1}} \cdots \sum_{a_{N-1}}B_{a_{N-n},a_{N-n+1}}^{(k_{N-n+1})} \cdots 
        B_{a_{N-1}}^{(k_N)} \ket{k_{N-n+1},\dots,k_N}.
    \end{equation}
\bm{$N-m.)$}\\
At the start of step $N-m$ we have
\begin{equation}
    \ketpsitwo = 
    \sum_{k_{m+2}} \cdots \sum_{k_{N}} \sum_{a_{m+1}} 
       \cdots \sum_{a_{N-1}}
       \lambda_{a_m,a_{m+1}} B_{a_{m+1},a_{m+2}}^{(k_{m+2})} \cdots
       B_{a_{N-1}}^{(k_N)} 
       \ket{x_{a_{m+1}}^{(m+1)}}\ket{k_{m+2},\dots,k_N}.
\end{equation}
We now expand $\ket{x_{a_{m+1}}^{(m+1)}}\in \hilb_{m+1}$ as $\sum_{k_{m+1}}x^{(k_{m+1})}_{a_{m+1}} \ket{k_{m+1}}$ to obtain
\begin{equation}
    \ketpsitwo = 
    \sum_{k_{m+1}} \cdots \sum_{k_{N}} \sum_{a_{m+1}} 
       \cdots \sum_{a_{N-1}}
       \lambda_{a_m,a_{m+1}} x^{(k_{m+1})}_{a_{m+1}} B_{a_{m+1},a_{m+2}}^{(k_{m+2})} \cdots 
       B_{a_{N-1}}^{(k_N)} 
       \ket{k_{m+1},\dots,k_N}.
\end{equation}
Defining $B^{(k_{m+1})}_{a_m,a_{m+1}} := \lambda_{a_m,a_{m+1}} x^{(k_{m+1})}_{a_{m+1}}$ yields the desired representation
\begin{equation} \label{mc-psitwo-}
    \ketpsitwo = 
    \sum_{k_{m+1}} \cdots \sum_{k_{N}} \sum_{a_{m+1}} 
       \cdots \sum_{a_{N-1}}
       B^{(k_{m+1})}_{a_m,a_{m+1}} B_{a_{m+1},a_{m+2}}^{(k_{m+2})} \cdots \\B_{a_{N-2},a_{N-1}}^{(k_{N-1})} B_{a_{N-1}}^{(k_N)} 
       \ket{k_{m+1},\dots,k_N},
\end{equation}
which converges to $\ketpsitwo$ in the norm of $\hilb_{m+1}\otimesdots\hilN.$
\\
\\
$ii.)$ \emph{Left leaf stripping.}
\\
Now we consider the vector $ \ket{\psi^{(1)}_{a_m}}$ in (\ref{mc-sd-m}) and decompose it in terms of left-normalized matrices. The left leaf stripping argument is similar to the right leaf stripping above, with the difference being in the direction of the procedure, going over the partitions $\lr{\hilone}\otimes\lr{\hiltwo\otimesdots\hilb_m},$
$\dots,\lr{\hilone\otimesdots\hilb_{m-1}}\otimes\lr{\hilb_m},$  and instead of the right Schmidt vectors, we now expand the left Schmidt vectors at each step to obtain the MPS matrix elements. The steps of the process are outlined below.
\\
\textbf{1.)}\\
We write a Schmidt decomposition of $\ketpsione$ with respect to the partition $(\hilone)\otimes(\hiltwo\otimes\cdots\otimes\hilb_m)$ and expand the left Schmidt vectors in the basis $\{\ket{k_1}\}$ as
\begin{equation}\label{lc_y1}
    \ket{x_{a_1}^{(1)}} = \sum_{k_1} A_{a_1}^{(k_1)}\ket{k_1}
\end{equation}
to obtain
\begin{align}
    \ketpsione 
    &= \sum_{k_1} 
    \sum_{a_1} A_{a_1}^{(k_1)} \lambda_{a_1,a_m} \ket{k_1} \ket{y_{a_1}^{(2,\dots, m)}}. \label{lc1}
\end{align}
The $a_m$-dependence in $\ketpsione$ is encoded in the Schmidt coefficients, and if $m=N,$ we can again set $a_m=1$ or omit it.
\\
\bm{$n=2,...,m-1.)$}\\ 
We Schmidt decompose $\ketpsione$ with respect to the partition $(\hilb_{1}\otimes\cdots\otimes\hilb_n)\otimes(\hilb_{n+1} \otimes\cdots\otimes \hilb_{m}),$ 
expand the left Schmidt vectors $\ket{x^{(1,\dots,n)}_{a_n}}$ in the basis $\{\ket{x_{a_{n-1}}^{(1,\dots,n-1)}}\ket{k_n}\}$ with coefficients denoted by $A_{a_{n-1},a_n}^{(k_n)}$ and 
write the vector $\ket{x_{a_{n-1}}^{(1,\dots,n-1)}}$ in the form obtained in the previous step\footnote{ For $n=2$ this is given by (\ref{lc_y1}), and for general $n$ we have\\ $\ket{x_{a_{n}}^{(1,\dots,n)}} = \sum_{k_{1}}\cdots\sum_{k_n} \sum_{a_1} \cdots\sum_{a_{n-1}}  A^{(k_1)}_{a_1} A^{(k_2)}_{a_{1},a_2} \cdots A^{(k_n)}_{a_{n-1},a_n} \ket{k_1,\dots,k_{n}}$
}
to obtain 
\begin{equation} \label{lc_n}
        \ketpsione
        = \sum_{k_1}\cdots\sum_{k_n} \sum_{a_1} \cdots \sum_{a_n} A_{a_1}^{(k_1)} A_{a_1,a_2}^{(k_2)} \cdots A_{a_{n-1},a_{n}}^{(k_n)}
        \lambda_{a_{n},a_m} 
        \ket{k_{1},\dots,k_{n}} \ket{y_{a_{n}}^{(n+1,\dots, m)}}.
\end{equation}
In particular, we obtained the following expression for the left Schmidt vectors:
\begin{equation} \label{lc_yn}
    \ket{x_{a_{n}}^{(1,\dots,n)}} = \sum_{k_{1}}\cdots\sum_{k_n} \sum_{a_1} \cdots\sum_{a_{n-1}}  A^{(k_1)}_{a_1} A^{(k_2)}_{a_{1},a_2} \cdots A^{(k_n)}_{a_{n-1},a_n} \ket{k_1,\dots,k_{n}}.
\end{equation}
$\bm{m.)}$\\
At the start of step $m$ we have 
\begin{equation}
	\ketpsione = \sum_{k_1} \cdots \sum_{k_{m-1}} 
	\sum_{a_1} \cdots \sum_{a_{m-1}}
	A_{a_1}^{(k_1)} A_{a_{1},a_{2}}^{(k_{2})} \cdots A_{a_{m-2}a_{m-1}}^{(k_{m-1})}
    \lambda_{a_{m-1},a_m} 
    \ket{k_1,\dots,k_{m-1}} \ket{ y_{a_{m-1}}^{(m)}}.
\end{equation}
Expanding $\ket{ y_{a_{m-1}}^{(m)}}$ 
as $\sum_{k_m}y_{a_{m-1}}^{(k_m)}\ket{k_N}$
and defining $A_{a_{m-1},a_m}^{(k_{m})} := \lambda_{a_{m-1},a_m}y_{a_{m-1}}^{(k_{m})}$ yields the desired norm-convergent representation
\begin{equation}\label{mc_psi1-}
	\ketpsione = \sum_{k_1=0}^\infty \cdots \sum_{k_m=0}^\infty 
	\sum_{a_1=0}^\infty \cdots \sum_{a_{m-1}=0}^\infty
	A_{a_1}^{(k_1)} A_{a_{1},a_{2}}^{(k_{2})} \cdots A_{a_{m-2}a_{m-1}}^{(k_{m-1})}A_{a_{m-1},a_m}^{(k_{m})}
	\ket{k_1,\dots,k_m}.
\end{equation}
$iii.)$ \emph{Combining the expressions.}\\
We combine (\ref{mc-sd-m}), (\ref{mc-psitwo-}) and (\ref{mc_psi1-}) and apply multilinearity and separate continuity of the tensor product to obtain the idMPS
\begin{equation} \label{mc_mps_final-}
    \ketpsi = 
    \sum_{k_1=0}^\infty \hspace{-1mm}\cdots \hspace{-1mm} \sum_{k_N=0}^\infty \sum_{a_1=0}^\infty \cdots \hspace{-2mm}\sum_{a_{N-1}=0}^\infty
    A^{(k_1)}_{a_1} \dots A^{(k_m)}_{a_{m-1},a_m}  D_{a_m,a_m}B_{a_m,a_{m+1}}^{(k_{m+1})} \cdots B_{a_{N-1}}^{(k_N)} \ket{k_1,\dots,k_N},
\end{equation}
where we defined a diagonal matrix $D_{a_m,a_m}:=\lambda_{a_m}$ containing the Schmidt coefficients of the partition $\lr{\hilone\otimesdots\hilb_m}\otimes\lr{\hilb_{m+1}\otimesdots\hilN}$ and reordered the sums. The series (\ref{mc_mps_final-}) converges to $\ketpsi$ in the norm of $\tpN.$ 
Thus the coefficients of $\ketpsi$ are represented by the series
\begin{equation}\label{RC-MPS-coeffs}
    c_{k_1,\dots,k_N} =\sum_{a_1=0}^\infty \cdots \hspace{-2mm}\sum_{a_{N-1}=0}^\infty
    A^{(k_1)}_{a_1} \dots A^{(k_m)}_{a_{m-1},a_m}  D_{a_m,a_m}B_{a_m,a_{m+1}}^{(k_{m+1})} \cdots B_{a_{N-1}}^{(k_N)}
\end{equation}
with convergence in $\cc.$
If $m=0,$ we obtain an MPS with only $B$-matrices exactly as in Definition \ref{RC-MPS}, that is,
\begin{equation}\label{RC-MPS-final}
    \ket{\psi} = \sum_{k_1=0}^\infty \cdots \sum_{k_N=0}^\infty 
    \sum_{a_1=0}^\infty \cdots \sum_{a_{N-1}=0}^\infty
    B_{a_{1}}^{(k_1)} B_{a_{1},a_{2}}^{(k_{2})} \cdots B_{a_{N-2},a_{N-1}}^{(k_{N-1})} B_{a_{N-1}}^{(k_N)} \ket{k_1,\dots,k_N}.
\end{equation}
Similarly, if $m=N,$ we obtain an MPS with only $A$-matrices as in Definition \ref{LC-MPS}:
\begin{equation}\label{LC-MPS-final}
	\ket{\psi} = \sum_{k_1=0}^\infty \cdots \sum_{k_N=0}^\infty 
	\sum_{a_1=0}^\infty \cdots \sum_{a_{N-1}=0}^\infty
	A_{a_1}^{(k_1)} A_{a_{1},a_{2}}^{(k_{2})} \cdots A_{a_{N-2}a_{N-1}}^{(k_{N-1})}A_{a_{N-1}}^{(k_{N})}
	\ket{k_1,\dots,k_N}.
\end{equation}
\emph{iv.) Normalization of the Matrices.}\\
To conclude the proof, we demonstrate that matrices $A^{(k_1)},\dots,A^{(k_N)}$ and $B^{(k_{1})},\dots,B^{(k_N)}$ constructed with the above methods are left- and right-normalized, respectively.

On sites $2$ to $N$, the right normalization formula reduces to an orthonormal inner product,
as seen by the following
(we introduce a dummy column index 1 to $B^{(k_N)}_{a_{N-1}} =: B^{(k_N)}_{a_{N-1},1}$):
\begin{align}
\sum_{k_n}\sum_{a_{n}} B^{(k_n)}_{a_{n-1},a_n} \left( B^{(k_n)}_{b_{n-1},a_n} \right)^*
= \braket{y_{b_{n-1}}^{(n,\dots,N)} | y_{a_{n-1}}^{(n,\dots,N)}} 
= \delta_{b_{n-1}a_{n-1}}. 
\end{align}
The condition also holds for the matrices 
$B_{a_{1}}^{(k_1)}$ 
as long as $\ketpsi$ is normalized, as demonstrated by the following calculation:
\begin{align}
    \sum_{k_1}\sum_{a_{1}} B^{(k_1)}_{a_{1}} \left( B^{(k_1)}_{a_{1}} \right)^*
    = \sum_{k_1}\sum_{a_{1}}  \lambda_{a_1}^2 x^{(k_1)}_{a_1}\left(x^{(k_1)}_{a_1}\right)^* 
    = \sum_{a_{1}}  \lambda_{a_1}^2 \braket{x_{a_1}^{(1)} | x_{a_1}^{(1)}} 
    = \sum_{a_{1}}  \lambda_{a_1}^2 
    = \norm{ \ketpsi }^2.
\end{align}

The argument for the left-normalization of the $A$-matrices is similar: on sites 1 to $N-1$ (introducing a dummy row index $a_0=1$)
the left-normalization condition reduces to the orthonormal inner product $\braket{x_{a_n}^{(1,\dots,n)} | x_{b_n}^{(1,\dots,n)}}.$ 
The condition also holds for the matrices $A_{a_{N-1}}^{(k_N)}$ 
as long as $\ketpsi$ is normalized, which can be deduced similarly as in the right-normalized case by using orthonormality of the left Schmidt vectors $\ket{y_{a_{N-1}}^{(N)}}$. 
\end{proof}

We can write the above construction more elegantly using tensor diagram notation.
We start with a general $N$-component tensor $\ketpsi\in\tpN$ and write its Schmidt decomposition with respect to the partition $(\hilone\otimesdots\hilb_m)\otimes(\hilb_{m+1}\otimesdots\hilN)$ as
\begin{equation} \label{sd-m-graph}
    \includegraphics[]{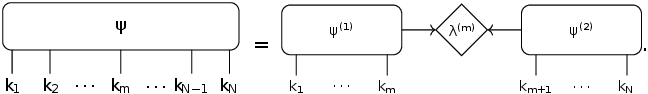}
\end{equation}
We apply right leaf stripping to $\ketpsitwo$ to write it as the directed graph
\begin{equation} \label{right-stripped-graph}
    \includegraphics[]{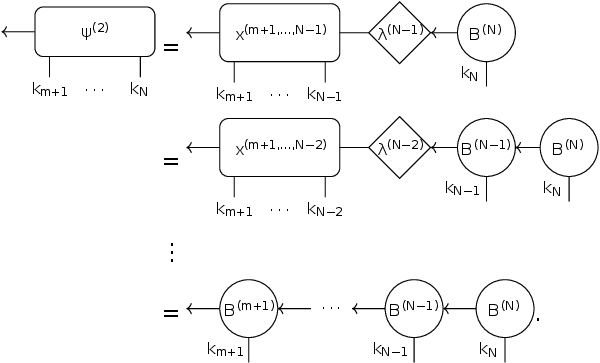}
\end{equation}
Similarly, we can apply left leaf stripping to write $\ketpsione$ as the directed graph
\begin{equation} \label{left-stripped-graph}
    \includegraphics[]{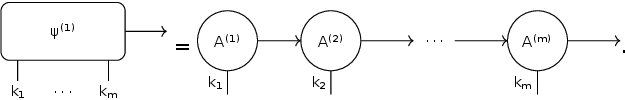}
\end{equation}
Combining (\ref{sd-m-graph}), (\ref{right-stripped-graph}) and (\ref{left-stripped-graph}) and denoting $\lambda^{(m)}=D$ yields the directed graph representing the MPS (\ref{mc_mps_final-})
\begin{equation}
    \includegraphics[]{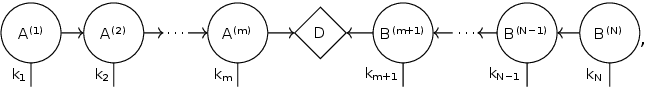}
\end{equation}
where the connected edges correspond to the (possibly infinite) summations over $a_1,\dots,a_{N-1},$ and the arrows denote the direction of the normalization of the matrices.

\subsection{Construction of Canonical idMPS}
\begin{proof}[Proof of Theorem \ref{idMPS existence} for the canonical form]
The construction of the canonical idMPS is also based on repeated Schmidt decompositions, and is analogous to the construction in \cite{vidal2003efficient}. The method is similar in spirit to the left leaf stripping method\footnote{We could equally well use a method analogous to right leaf stripping with the same basic ideas.} introduced in the mixed canonical construction in Section 3.1, with differences in details that let us leave the Schmidt coefficients explicitly visible in the final MPS.
Let us outline the construction for a general element $\ketpsi\in\tpN$ of the form (\ref{ketpsigeneral}). \\
$\bm{1.)}$\\
    We first Schmidt decompose $\ket{\psi}$ with respect to the partition $\hilone \otimes \big( \hiltwo \otimes \cdots \otimes \hilN \big)$
    and expand the left Schmidt vectors $\ket{x\superone_{a_1}} \in \hilone$ in the basis $\{ \ket{k_1} \}$ with coefficients $\Gamma_{a_1}^{(k_1)} \in \cc$ 
    to obtain
    \begin{equation}\label{c3}
        \ket{\psi} = \sum_{k_1=0}^\infty \sum_{a_1=0}^\infty \Gamma_{a_1}^{(k_1)} \lambda^{(1)}_{a_1}\ket{k_1} \ket{y^{(2,...,N)}_{a_1}},
    \end{equation}
    which converges in norm.
    \\
$\bm{2.)}$
\vspace{-3mm}
\begin{enumerate}[label=(\roman*),font=\itshape]
    \item
    For each $a_1,$ we write $\ket{y_{a_1}^{(2,...,N)}} \in \hiltwo \otimes\cdots \otimes \hilN$ in the basis $\{ \ket{k_2,\dots,k_N} \}$ as
    \begin{align}
        \ket{y_{a_1}^{(2,...,N)}} 
        &= \sum_{k_2}\cdots\sum_{k_N} y_{a_1}^{(k_2,...,k_N)} \ket{k_2}\ket{k_3,\dots,k_N}\\
        &=: \sum_{k_2} \ket{k_2} \ket{\chi^{(3,...,N)}_{a_1,k_2} }, \label{c4}
    \end{align}
    where we applied multilinearity and separate continuity to define
    \begin{equation}
        \ket{\chi^{(3,...,N)}_{a_1,k_2}} := \sum_{k_3} \cdots \sum_{k_N} y_{a_1}^{(k_2,...,k_N)}  \ket{k_3,\dots,k_N}
	\in \hilbert_3 \otimes \dots \otimes \hilN.
    \end{equation}
\item 
    Schmidt decompose $\ket{\psi}$ with respect to the partition $(\hilone\otimes \hilbert_2)\otimes(\hilbert_3 \otimes \dots \otimes \hilN)$: 
    \begin{equation}  \label{c5}
        \ket{\psi}
        = \sum_{a_2=0}^\infty \lambda^{(2)}_{a_2} \ket{x^{(1,2)}_{a_2}}  \ket{y^{(3,...,N)}_{a_2} }.
    \end{equation}
	Here $\{ \ket{y^{(3,...,N)}_{a_2} } \}$ can be extended into an orthonormal basis for $\hilbert_3 \otimes \dots \otimes\hilbert_N$ and thus we can write 
	\begin{equation}  \label{c6}
		\ket{\chi^{(3,...,N)}_{a_1,k_2} } = \sum_{a_2=0}^\infty \tau_{a_1,a_2}^{(k_2)}  \ket{y^{(3,...,N)}_{a_2} }
		= \sum_{a_2=0}^\infty \Gamma^{(k_2)}_{a_1,a_2} \lambda^{(2)}_{a_2}  \ket{y^{(3,...,N)}_{a_2} },  
	\end{equation}
	where in the last equality we wrote the tensor coefficients in terms of the Schmidt coefficients as $\tau_{a_1,a_2}^{(k_2)} = \Gamma^{(k_2)}_{a_1,a_2} \lambda^{(2)}_{a_2}.$ Let us quickly justify why this can be done. If for some $a_2$ we have $\lambda^{(2)}_{a_2}\neq0,$ then we can set $\Gamma^{(k_2)}_{a_1,a_2} = \tau_{a_1,a_2}^{(k_2)} / \lambda^{(2)}_{a_2}.$ For the case $\lambda^{(2)}_{a_2}=0,$ notice first that there exist $\beta_{a_2}^{(k_1,k_2)}\in\cc$ such that
    \begin{equation}  \label{c7}
        \ket{x^{(1,2)}_{a_2}} = \sum_{k_1} \sum_{k_2} \beta_{a_2}^{(k_1,k_2)}\ket{k_1, k_2}.
    \end{equation}
    Now by combining (\ref{c3}), (\ref{c4}) and (\ref{c6}) on the left-hand side as well as (\ref{c5}) and (\ref{c7}) on the right-hand side we obtain 
    \begin{align}
        & \sum_{k_1=0}^{d_1-1}\sum_{k_2=0}^{d_2-1}\sum_{a_1=0}^\infty \sum_{a_2=0}^\infty 
        \Gamma^{(k_1)}_{a_1}\lambda^{(1)}_{a_1}\tau^{(k_2)}_{a_1,a_2} \ket{k_1, k_2} \ket{y^{(3,\dots,N)}_{a_2}} \\
        &= \sum_{k_1=0}^{d_1-1}\sum_{k_2=0}^{d_2-1} \sum_{a_2=0}^\infty 
        \lambda^{(2)}_{a_2} \beta_{a_2}^{(k_1,k_2)} \ket{k_1, k_2} \ket{y^{(3,\dots,N)}_{a_2}},
    \end{align}
    which implies that
    \begin{equation}
        \sum_{a_1=0}^\infty \Gamma^{(k_1)}_{a_1}\lambda^{(1)}_{a_1}\tau^{(k_2)}_{a_1,a_2}
        = \lambda^{(2)}_{a_2} \beta_{a_2}^{(k_1,k_2)}.
    \end{equation}
    Now if for some $a_2$ we have $\lambda^{(2)}_{a_2}=0,$ then also $\tau^{(k_2)}_{a_1,a_2}=0$ for every $a_1\in\nn.$ In this case we set $\Gamma^{(k_2)}_{a_1,a_2} = \tau_{a_1,a_2}^{(k_2)}.$
	\item
	Substitute (\ref{c6}) to (\ref{c4}) to (\ref{c3})
 and rearrange the summations to obtain
	\begin{equation}
		\ket{\psi} = 
		\sum_{k_2}\sum_{k_2} \sum_{a_1} \sum_{a_2}
		\Gamma_{a_1}^{(k_1)} \lambda^{(1)}_{a_1}  \Gamma^{(k_2)}_{a_1,a_2} \lambda^{(2)}_{a_2} \ket{k_1, k_2} \ket{y^{(3,...,N)}_{a_2}}.
	\end{equation}
\end{enumerate}
$\bm{3,...,N-1.)}$\\
    We continue this procedure iteratively, repeating steps $(i)$, $(ii)$ and $(iii)$ (with the obvious modifications) for the Schmidt vectors $\ket{y^{(3,\dots,N)}_{a_2}}, \dots ,\ket{y^{(N-1,N)}_{a_{N-2}}} $ until after $N-1$ Schmidt decompositions we obtain an expression of the form
    \begin{equation} \label{c8}
	\ket{\psi} = 
	\sum_{k_1}\cdots\sum_{k_{N-1}} \sum_{a_1}\cdots\sum_{a_{N-1}}
        \Gamma_{a_1}^{(k_1)} \lambda^{(1)}_{a_1}  \Gamma^{(k_2)}_{a_1,a_2} \lambda^{(2)}_{a_2} 
	\cdots 
	\Gamma^{(k_{N-1})}_{a_{N-2},a_{N-1}} \lambda^{(N-1)}_{a_{N-1}} 
	\ket{k_1,\dots,k_{N-1}} \ket{y^{(N)}_{a_{N-1}}}.
    \end{equation}
		\\
$\bm{N.)}$\\
    As the final step, we expand the vectors $\ket{y^{(N)}_{a_{N-1}}} \in \hilN$ in the basis $\{ \ket{k_N} \}$ as
		\begin{equation}  \label{c9}
			\ket{y^{(N)}_{a_{N-1}}} = \sum_{k_N} \Gamma^{(k_N)}_{a_{N-1}} \ket{k_N}.
		\end{equation}
    Substituting (\ref{c9}) to (\ref{c8}), reordering the sums and applying multilinearity and separate continuity yields the expression
    \begin{equation} \label{C-MPS_final}
	\ket{\psi} =
	 \sum_{k_1=0}^\infty\cdots\sum_{k_N=0}^\infty \sum_{a_1=0}^\infty\cdots\sum_{a_{N-1}=0}^\infty
        \Gamma_{a_1}^{(k_1)} \, \lambda^{(1)}_{a_1}  \Gamma^{(k_2)}_{a_1,a_2} \, \lambda^{(2)}_{a_2} 
	\cdots 
        \Gamma^{(k_{N-1})}_{a_{N-2},a_{N-1}} \, \lambda^{(N-1)}_{a_{N-1}}  \Gamma^{(k_{N})}_{a_{N-1}}
	\ket{k_1,\dots,k_{N}}\hspace{-1mm},
    \end{equation}
    with convergence in $\hilone \otimes \cdots \otimes \hilN.$ Thus we have
        \begin{equation}
	c_{k_1,\dots,k_N} = 
	 \sum_{a_1=0}^\infty \cdots \sum_{a_{N-1}=0}^\infty
        \Gamma_{a_1}^{(k_1)} \, \lambda^{(1)}_{a_1}  \Gamma^{(k_2)}_{a_1,a_2} \, \lambda^{(2)}_{a_2} 
	\cdots 
        \Gamma^{(k_{N-1})}_{a_{N-2},a_{N-1}} \, \lambda^{(N-1)}_{a_{N-1}}  \Gamma^{(k_{N})}_{a_{N-1}},
    \end{equation}
    with convergence in $\cc.$
    The Proposition below concludes that (\ref{C-MPS_final}) is indeed a canonical MPS.
\end{proof}

\begin{prop} \label{SD from canonical MPS}
The MPS given in (\ref{C-MPS_final}) is canonical in the sense of Definition \ref{C-MPS}.
\end{prop}
\begin{proof}
    By construction, for every $n\in\{ 1,...,N-2 \}$ we have the following expressions for the Schmidt vectors (when taking the Schmidt decomposition with respect to the partition $(\hilone\otimes\cdots\otimes\hilb_n)\otimes(\hilb_{n+1}\otimes\cdots\otimes\hilN)$):
    \begin{align} 
         \ket{x^{(1,...,n)}_{a_n}} &= \sum_{k_1,\dots,k_{n}} \sum_{a_1,...,a_{n-1}}
        \Gamma_{a_1}^{(k_1)} \, \lambda^{(1)}_{a_1}
	   \cdots 
        \lambda^{(n-1)}_{a_{n-1}}  \Gamma^{(k_{n})}_{a_{n-1},a_n}
	   \ket{k_1,\dots,k_{n}},
	   \label{csd1}
    \\
        \ket{y^{(n+1,...,N)}_{a_n}} &= \sum_{k_{n+1}} \sum_{a_{n+1}}
        \Gamma_{a_n,a_{n+1}}^{(k_{n+1})} \, \lambda^{(n+1)}_{a_{n+1}}
	   \ket{k_{n+1}}\ket{y^{(n+2,...,N)}_{a_{n+1}}},
	   \label{csd2}
    \\
        \ket{y^{(N)}_{a_{N-1}}} &= \sum_{k_{N}} \Gamma_{a_{N-1}}^{(k_N)} \ket{k_N}.
        \label{csd3}
    \end{align}
    Based on (\ref{csd2}) and (\ref{csd3}), we can express all of the Schmidt vectors $\ket{y^{(n+1,...,N)}_{a_n}}$ in the form
    \begin{align}
        \ket{y^{(n+1,...,N)}_{a_n}} &= 
        \sum_{k_{n+1}} \sum_{a_{n+1}}
        \Gamma_{a_n,a_{n+1}}^{(k_{n+1})} \, \lambda^{(n+1)}_{a_{n+1}}
	   \ket{k_{n+1}}\ket{y^{(n+2,...,N)}_{a_{n+1}}} \nonumber
        \\
        &=
        \sum_{k_{n+1}}\sum_{k_{n+2}}\sum_{a_{n+1}}\sum_{a_{n+2}}
        \Gamma_{a_n,a_{n+1}}^{(k_{n+1})} \, \lambda^{(n+1)}_{a_{n+1}}
        \Gamma_{a_n,a_{n+2}}^{(k_{n+2})} \, \lambda^{(n+2)}_{a_{n+2}}
        \ket{k_{n+1},k_{n+2}}\ket{y^{(n+3,...,N)}_{a_{n+2}}} \nonumber
        \\
        & \;\: \vdots \nonumber
        \\
        &=\sum_{k_{n+1}} \hspace{-0.3mm} \cdots \hspace{-0.3mm} \sum_{k_N} \sum_{a_{n+1}}\hspace{-0.3mm} \cdots \hspace{-0.3mm}\sum_{a_{N-1}}         \Gamma_{a_n,a_{n+1}}^{(k_{n+1})} \, \lambda^{(n+1)}_{a_{n+1}} \cdots  \lambda^{(N-1)}_{a_{N-1}}  \Gamma^{(k_{N})}_{a_{N-1},a_N} \ket{k_{n+1},\dots,k_N},
    \end{align}
    as desired. Let us check the partition $(\hilone\otimes\cdots\otimes\hilb_{N-1})\otimes(\hilN)$ separately. The equality 
    \begin{align} \label{csd4}
        \ket{\psi} =& 
        \sum_{a_{N-1}}
        \lambda_{a_{N-1}}^{(N-1)} 
        \left( 
        \sum_{k_1,\dots,k_{N-1}} \sum_{a_1,...,a_{N-2}}
        \Gamma_{a_1}^{(k_1)} \, \lambda^{(1)}_{a_1}
	   \cdots 
        \lambda^{(N-2)}_{a_{N-2}}  \Gamma^{(k_{N-1})}_{a_{N-2},a_{N-1}}
	   \ket{k_1,\dots,k_{N-1}}
        \right) \nonumber
        \\
        & \otimes
         \left( 
        \sum_{k_{N}}
        \Gamma^{(k_{N})}_{a_{N-1}}\ket{k_{N}}
        \right),
    \end{align}
    obviously holds with convergence in the norm of $\tpN$ (this can be seen by applying multilinearity and separate continuity of the tensor product).
    Additionally, by (\ref{csd3}) we have   $\ket{y^{(N)}_{a_{N-1}}} = \sum_{k_{N}} \Gamma_{a_{N-1}}^{(k_N)} \ket{k_N}.$ Thus the sum (\ref{csd4}) has the Schmidt coefficients and the right Schmidt vectors $\ket{y^{(N)}_{a_{N-1}}}$, and we deduce that necessarily 
    \footnote{Because in general if $\{\ket{\phi_k}\}$ is an orthonormal set and the equality
    $ \sumk \ket{\xi_k}\otimes\ket{\phi_k} = \sumk \ket{\chi_k}\otimes\ket{\phi_k} $
    holds, 
    then $\ket{\xi_k} = \ket{\chi_k}$ for every $k\in\nn.$}  
    \begin{equation}
         \sum_{k_1}\hspace{-1mm}\cdots\hspace{-1mm}\sum_{k_{N-1}}\hspace{-1mm} \sum_{a_1}\hspace{-1mm}\cdots\hspace{-1mm}\sum_{a_{N-2}}
        \Gamma_{a_1}^{(k_1)} \, \lambda^{(1)}_{a_1}
	   \cdots 
        \lambda^{(N-2)}_{a_{N-2}}  \Gamma^{(k_{N-1})}_{a_{N-2},a_{N-1}}
	   \ket{k_1,\dots,k_{N-1}} = \ket{x^{(1,...,N-1)}_{a_{N-1}}}.       
    \end{equation}
    Therefore we have a valid Schmidt decomposition and thus the MPS (\ref{C-MPS_final}) is canonical.
\end{proof}
Using tensor diagrams, the above construction can be written as 
\begin{equation}
    \includegraphics[]{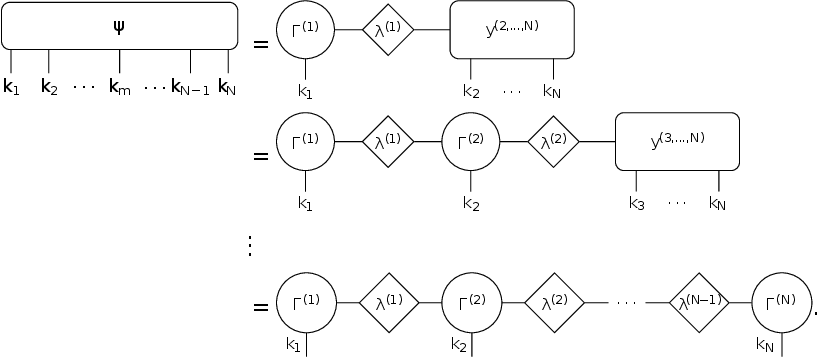}
\end{equation}

\subsection{Interpretation of idMPS as a Product of Operators}
As the construction of idMPS is exactly analogous to the finite-dimensional case, idMPS inherit the properties of regular MPS. However, as the bond dimension may be infinite, the matrices now give rise to operators on possibly infinite-dimensional (auxiliary) Hilbert spaces, and it is interesting to study the properties of these operators. For example, if we could show that they are compact (under some assumptions), then each of them could be individually approximated by finite-rank operators.

Let us explain how we can interpret an infinite-dimensional MPS as a composition of operators. Consider a left-canonical three-particle MPS given by
    \begin{equation}
        \ket{\psi} = \sum_{k_1} \sum_{k_2} \sum_{k_3} \sum_{a_1} \sum_{a_{2}}
	A_{a_1}^{(k_1)} A_{a_{1},a_{2}}^{(k_{2})} A_{a_{2}}^{(k_{3})}
	\ket{k_1,k_2,k_3}.
    \end{equation}
    The generalization of the following to the general $N$-particle case is straightforward.
    
    We would like to be able to identify $A_{a_{2}}^{(k_{3})}$ with an $\ell^2$ sequence, $A_{a_1}^{(k_1)}$ with a functional in $(\ell^2)^*$ and $A_{a_{1},a_{2}}^{(k_{2})}$ with an operator acting on $\ell^2.$

    For a fixed value of $k_1,$ we can define a functional $T^{(1,k_1)}$ acting on $\ell^2$ according to the formula
    \begin{equation}
        T^{(1,k_1)} x := \sum_{a_1} A_{a_1}^{(k_1)} x_{a_1},
    \end{equation}
    where $x=(x_{a_1})_{a_1\in\nn} \in \ell^2.$

    For a fixed value of $k_2,$ we can define an operator $T^{(2,k_2)}$ acting on $\ell^2$ such that
    \begin{equation}
        \big( T^{(2,k_2)} x \big)_n := \sum_{a_2} A^{(k_2)}_{n,a_2} x_{a_2},
    \end{equation}
    where $x=(x_{a_2})_{a_2\in\nn} \in \ell^2.$

    Finally, for a fixed value of $k_3,$ we can define a sequence $(x_n)_{n\in\nn}$ by $x_n := A^{(k_3)}_{n}.$ To show that $(x_n)\in\ell^2,$ we would have to prove that for any fixed $k_3$ it holds that $\sum_{n=0}^\infty |A^{(k_3)}_{n}|^2 < \infty.$

\section{MPS for a Chain of Three Coupled Harmonic Oscillators}
\subsection{The Problem}

As an application of the previous results, in this section we construct an idMPS expression for certain eigenstates of a chain of three coupled harmonic oscillators.
The system under consideration is governed by the Hamiltonian 
\begin{equation}
    H = \frac{1}{2}\left(\sum_{i=1}^{3}\frac{p_i^2}{m_i} +m_i\omega_i^2x_i^2 \right) + D_{12}x_1x_2 + D_{13}x_1x_3 + D_{23}x_2x_3.
\end{equation}

It is demonstrated in \cite{merdaci2020entanglement} and \cite{merdaci2024entanglement} that under certain assumptions (see equations (3)-(5) in \cite{merdaci2024entanglement}) the eigenstates of this system can be written in the form 
\begin{align}
    \psi^{ABC}_{n_1,n_2,n_3}(x_1,x_2,x_3)
    =& 
    \frac{\left( \frac{m \Tilde{\omega} }{\pi \hbar} \right)^{3/4}}{\sqrt{n_1!n_2!n_3!2^{n_1+n_2+n_3}}} e^{-\frac{m\Tilde{\omega}}{2\hbar} (q_1^2+q_2^2+q_3^2) } \nonumber \\
    & \;\; \hermite{n_1}\left( q_1\sqrt{ \frac{m\Tilde{\omega}}{\hbar}} \right) 
    \hermite{n_2}\left( q_2 \sqrt{ \frac{m\Tilde{\omega}}{\hbar}} \right)
    \hermite{n_3}\left( q_3 \sqrt{\frac{m\Tilde{\omega}}{\hbar}} \right),
\end{align}
where $\hermite{n_i}$ are the (physicist's) Hermite polynomials, and explicit expressions for the parameters $\omegatilde$ and $m$ as well as the coordinates $q_i$ in terms of $x_i$ are given in Appendix A of \cite{merdaci2024entanglement}.

Setting $m=1$ and $\hbar=1$ and considering the special case $n_1=n_2=0$ yields the simplified expression
\begin{align}
    \psi^{ABC}_{0,0,n_3}(x_1,x_2,x_3)
    =& 
    \frac{\left( \Tilde{\omega}/\pi \right)^{3/4}}{\sqrt{
    n_3!2^{n_3}}} e^{-\frac{\Tilde{\omega}}{2} (q_1^2+q_2^2+q_3^2) } 
    \hermite{n_3}\left( q_3 \sqrt{\Tilde{\omega}} \right).
\end{align}

\subsection{Constructing the MPS Representation}
In this section we derive a left-canonical MPS representation for the eigenstates with $n_1=0$ and $n_2=0$ in the unscaled bases $\{ f^{(i)}_k(x_i) \}\subseteq \hili,$ where 
\begin{equation} \label{oscillator_single_particle_basis}
    f^{(i)}_k(x_i) = \frac{1}{\pi^{1/4}\sqrt{2^k k!}} e^{-x_i^2/2} \hermite{k} \left(x_i \right).
\end{equation}
The necessary Schmidt decompositions are derived in \cite{merdaci2024entanglement} (equations (36), (38), (52), (93) thereof), and are given by
\begin{align}
    \psi_{0,0,n}(x_1,x_2,x_3) &= \sum_{a=0}^n \sqrt{\alpha_a} \: \varphi^A_{a}(x_1) \Theta^{BC}_a(x_2,x_3) \in \hilone \otimes \left( \hiltwo \otimes \hilb_3 \right) \label{oscillator-sd-1:23} \\
    &= \sum_{b=0}^n \sqrt{\gamma_b} \: \Xi^{AB}_{b}(x_1,x_2) \chi^{C}_b(x_3) \in \left( \hilone \otimes \hilb_2 \right) \otimes \hilb_3, \label{oscillator-sd-12:3}
\end{align}
where
\begin{align}
    \varphi^A_{a}(x_1) &= \left( \frac{ \sqrt{\omegatilde} }{\sqrt{\pi}2^a a!  } \right)^{ \frac{1}{2} }
    e^{-\omegatilde x_1^2/2}
    \hermite{a}\left( \sqrt{\omegatilde} x_1 \right), \\
    \phi^B_{l}(x_2) &= \left( \frac{ \sqrt{\omegatilde} }{\sqrt{\pi}2^l l!  } \right)^{ \frac{1}{2} } e^{-\omegatilde x_2^2/2} \hermite{l}\left( \sqrt{\omegatilde} x_2 \right), \\
    \chi^{C}_b(x_3) &= \left( \frac{ \sqrt{\omegatilde} }{\sqrt{\pi}2^b b!  } \right)^{ \frac{1}{2} } e^{-\omegatilde x_3^2/2} \hermite{b}\left( \sqrt{\omegatilde} x_3 \right), \\
    \Theta^{BC}_a(x_2,x_3) &= \sum_{l=0}^{n-a} \phi^B_{l}(x_2) \chi^{C}_{n-a-l} (x_3), \\ 
    \Xi^{AB}_{b}(x_1,x_2) &= \sum_{k=0}^{n-b} \varphi^A_{k}(x_1) \phi^B_{n-k-b}(x_2), \\ 
    \alpha_a &= \frac{n!}{a!(n-a)!} \sin^{2a}\theta \cos^{2a}\phi \left( 
1- \sin^2\theta \cos^2\phi \right)^{n-a}, \\ 
    \gamma_b &= \frac{n!}{b!(n-b)!} \left( \cos\theta\cos\varphi + \sin\theta\sin\phi\sin\varphi \right)^{2b} \nonumber \\ 
    & \;\;\;\; \left[ (\cos\theta\sin\varphi - \sin\theta\cos\varphi\sin\phi)^2 + \cos^2\phi\sin^2\theta \right]^{n-b}. 
\end{align}

In the construction of the MPS we encounter the integral
\begin{equation}
    I_{i,j} = \int_{-\infty}^\infty e^{-\frac{1}{2}( x^2 + \Tilde{\omega} x^2 )} \hermite{i}(x)\hermite{j}(\sqrt{\Tilde{\omega}}\,x)dx,
\end{equation}
for which we compute a closed form expression in Appendix \ref{Appendix B}.

At this point one can notice that constructing the MPS in the basis $\{ \varphi^A_{k}(x_1) \phi^B_{l}(x_2) \chi^{C}_m(x_3) \}$ would yield a trivial MPS with two identity matrices and a third with Schmidt coefficients on the diagonal. Let us proceed in the basis $\{f^{(1)}_k(x_1) f^{(2)}_l(x_2) f^{(3)}_m(x_3) \}.$
\\\\
\textbf{1.)} 
Let us first expand the Schmidt vectors $\varphi^A_{a}(x_1)$ and $\Theta^{BC}_a(x_2,x_3)$ in terms of the basis vectors $f^{(1)}_k(x_1)$ and $f^{(2)}_l(x_2) f^{(3)}_m(x_3),$ respectively. Defining first
\begin{equation}
    C_{i,j} =  \sqrt{\frac{\sqrt{\Tilde{\omega}}}{\pi 2^i 2^j i! j!}},
\end{equation}
we obtain
\begin{align}
    \varphi^A_{a}(x_1) = \sum_{k=0}^\infty \braket{f^{(1)}_k | \varphi^A_{a}} f^{(1)}_k(x_1) 
    = \sum_{k=0}^\infty C_{k,a} I_{k,a} f^{(1)}_k(x_1), \label{oscillator-mps-sd1}
\end{align}
and
\begin{align}
        \Theta^{BC}_a(x_2,x_3) =& \sum_{l=0}^\infty \sum_{m=0}^\infty \braket{f^{(2)}_l f^{(3)}_m | \Theta^{BC}_a} f^{(2)}_l(x_2)f^{(3)}_m(x_3) \nonumber \\
         =& \sum_{l=0}^\infty \sum_{m=0}^\infty \sum_{l'=0}^{n-a} 
         C_{l,l'} C_{m,n-a-l'}
         I_{l,l'} I_{m,n-a-l'} f^{(2)}_l(x_2)f^{(3)}_m(x_3).
         \label{oscillator-mps-sd2}
\end{align}
Combining (\ref{oscillator-sd-1:23}), (\ref{oscillator-mps-sd1}) and (\ref{oscillator-mps-sd2}) yields the expression
\begin{align}
    \psi_{0,0,n}(x_1,x_2,x_3) 
    = \sum_{k,l,m=0}^\infty \sum_{a=0}^n &
    \sqrt{\alpha_a}  C_{k,a} I_{k,a} 
    \left( \sum_{l'=0}^{n-a} C_{l,l'} C_{m,n-a-l'} I_{l,l'} I_{m,n-a-l'} \right) \nonumber\\
    &f^{(1)}_k(x_1)f^{(2)}_l(x_2)f^{(3)}_m(x_3) \nonumber \\
    = \sum_{k,l,m=0}^\infty\sum_{a=0}^n & A^{(1,k)}_a \sqrt{\alpha_a}
    \left( \sum_{l'=0}^{n-a} C_{l,l'} C_{m,n-a-l'} I_{l,l'} I_{m,n-a-l'} \right)
    \nonumber\\
    &f^{(1)}_k(x_1)f^{(2)}_l(x_2)f^{(3)}_m(x_3),
\end{align}
where we denoted 
\begin{equation} \label{oscillator-mps-A1}
   A^{(1,k)}_a = C_{k,a} I_{k,a},
\end{equation}
writing the site index $1$ explicitly to avoid confusion.
\\\\
\textbf{2.)} Let us now expand the Schmidt vectors $\Xi^{AB}_{b}(x_1,x_2) $ and $\chi^{C}_b(x_3)$ in terms of the basis vectors $\varphi^{A}_a(x_1)f^{(2)}_l(x_2)$ and $f^{(3)}_m(x_3),$ respectively. Proceeding as in step 1 we obtain
\begin{align}
    \Xi^{AB}_{b}(x_1,x_2) 
    &= \sum_{a=0}^{n} \sum_{l=0}^\infty \indic_{\{a \leq n-b\}} C_{l,n-a-b} I_{l,n-a-b} \varphi^{A}_a(x_1)f^{(2)}_l(x_2) \label{oscillator-mps-sd3}
\end{align}
and 
\begin{align}
    \chi^{C}_b(x_3) = \sum_{m=0}^\infty C_{m,b}I_{m,b} f^{(3)}_m(x_3),
    \label{oscillator-mps-sd4}
\end{align}
where $C_{i,j}$ is as previously and $\indic_{\{a \leq n-b\}}$ is the indicator function. Now combining (\ref{oscillator-sd-12:3}), (\ref{oscillator-mps-sd3}) and (\ref{oscillator-mps-sd4}) yields the expression 
\begin{align}
    \psi_{0,0,n}(x_1,x_2,x_3) \hspace{-1mm} &= \hspace{-1mm} \sum_{l=0}^\infty \hspace{-0.3mm} \sum_{m=0}^\infty \hspace{-0.3mm} \sum_{a=0}^{n} \hspace{-0.3mm} \sum_{b=0}^{n} \hspace{-0.4mm} \sqrt{\gamma_b} 
    \indic_{\{a+b \leq n\}} C_{l,n-a-b} I_{l,n-a-b} C_{m,b}I_{m,b}
    \varphi^{A}_a(x_1) f^{(2)}_l\hspace{-0.4mm}(x_2) f^{(3)}_m\hspace{-0.4mm}(x_3).
\end{align}
Expanding $\varphi^{A}_a(x_1)$ in terms of $f^{(1)}_k$ as in step 1 yields
\begin{align}
    \psi_{0,0,n}(x_1,x_2,x_3) = \sum_{k=0}^\infty \sum_{l=0}^\infty \sum_{m=0}^\infty \sum_{a=0}^{n} \sum_{b=0}^{n} & A^{(1,k)}_a \sqrt{\gamma_b} 
    \indic_{\{a+b \leq n\}} C_{l,n-a-b} I_{l,n-a-b} C_{m,b}I_{m,b}
    \nonumber \\
    & f^{(1)}_k(x_1) f^{(2)}_l(x_2) f^{(3)}_m(x_3).
\end{align}
Denoting
\begin{equation} \label{oscillator-mps-A2}
    A^{(2,l)}_{a,b} =\indic_{\{a+b \leq n\}}C_{l,n-a-b} I_{l,n-a-b}
\end{equation}
and
\begin{equation} \label{oscillator-mps-A3}
    A^{(3,m)}_b = \sqrt{\gamma_b} C_{m,b}I_{m,b}
\end{equation}
yields the MPS
\begin{align} 
    \psi_{0,0,n}(x_1,x_2,x_3) &= \sum_{k=0}^\infty \sum_{l=0}^\infty \sum_{m=0}^\infty \sum_{a=0}^{n} \sum_{b=0}^{n} A^{(1,k)}_a A^{(2,l)}_{a,b} A^{(3,m)}_b
    f^{(1)}_k(x_1) f^{(2)}_l(x_2) f^{(3)}_m(x_3). \label{oscillator-mps}
\end{align}
Thus we obtained an MPS of finite bond dimension, where each coefficient of the wavefunction in our chosen basis is given as a product of three matrices.

\begin{figure}
\includegraphics[width=.33\linewidth]{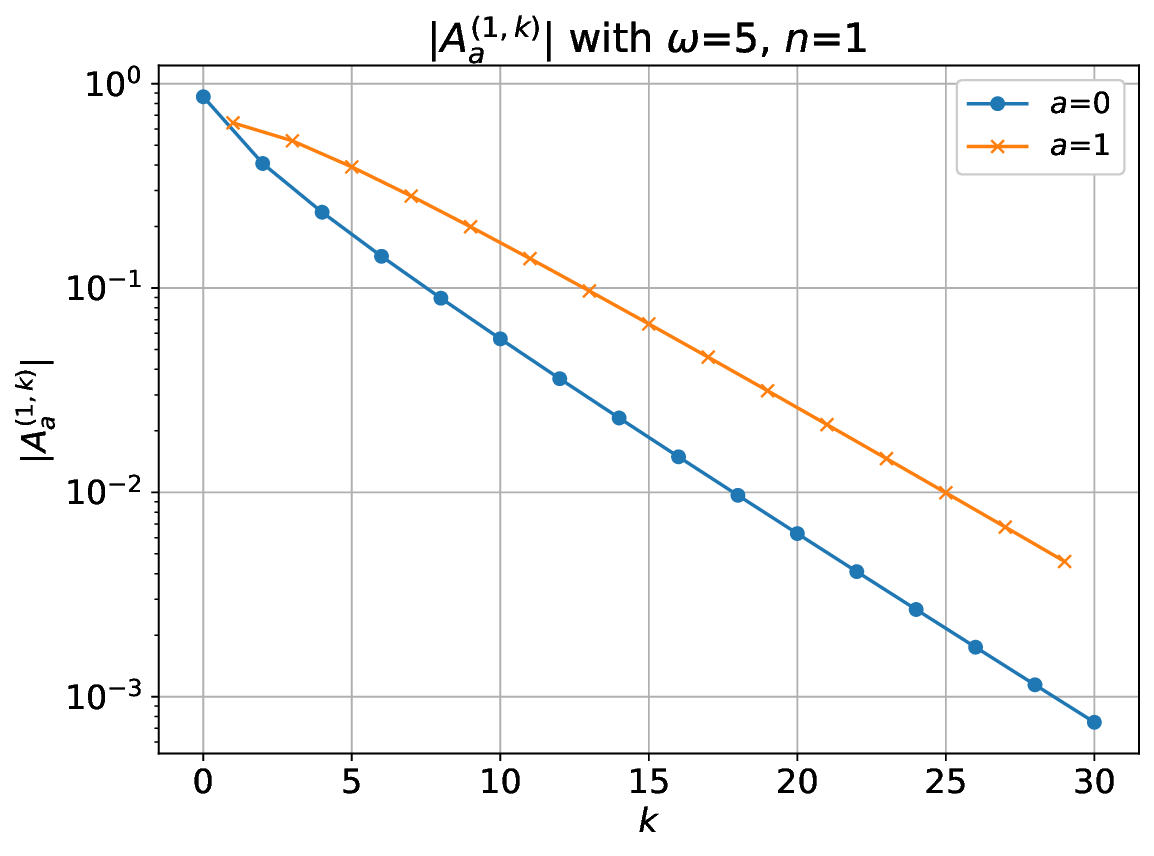}\hfill
\includegraphics[width=.33\linewidth]{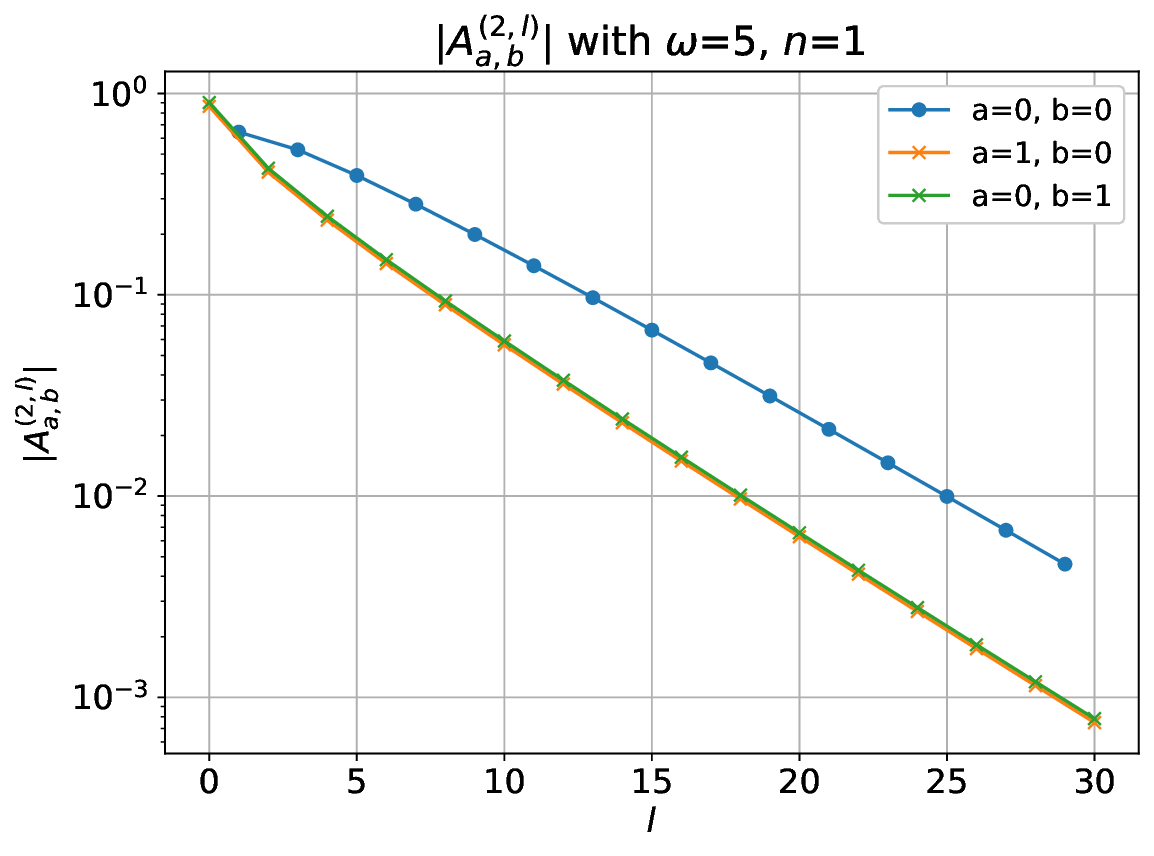}
\includegraphics[width=.33\linewidth]{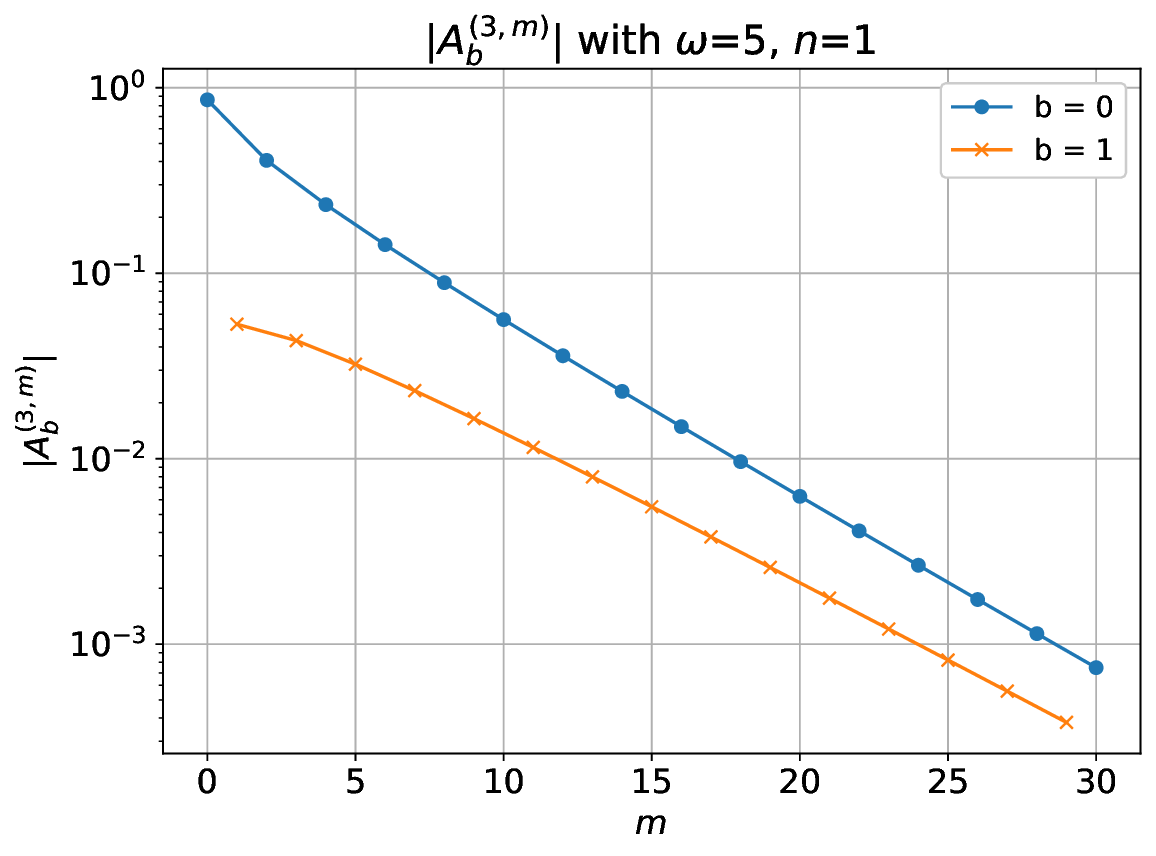}
\caption{(color online) Absolute values of the nonzero MPS matrix elements for the first excited state ($n=1$) as functions of the physical indices with $\omega=5$ and $D_{12}=D_{23}=0.25,\\ D_{13}=0.$ They are monotone decreasing, and exponentially decaying after a certain point.}
\end{figure}

In Figures 1 and 2 we have plotted numerical values of the nonzero absolute values of the matrix elements $A^{(1,k)}_a, A^{(2,l)}_{a,b}$ and $A^{(3,m)}_b$ as functions of the physical indices $k,l$ and $m,$ respectively, for the first two excited states of the system with nearest-neighbor interactions and for certain values of the physical parameters. The matrices $A^{(1,k)}$ and $A^{(2,l)}$ share the same elements and $A^{(2,l)}$ is symmetric. Both of these facts can be directly deduced from (\ref{oscillator-mps-A1}) and (\ref{oscillator-mps-A2}). After a certain point (which depends on the quantum number), all of the matrix elements decay exponentially.

\section{Conclusion and Outlook}

We demonstrated that any element in the tensor product of separable infinite-dimensional Hilbert spaces can be expressed as an infinite-dimensional matrix product state (idMPS) in any of the canonical forms. The existence of the Schmidt decomposition in general separable Hilbert spaces
allowed us to construct idMPS in a manner completely analogous to the already well-established finite-dimensional case. Therefore idMPS inherit many of the desirable properties associated with finite-dimensional MPS. 

\begin{figure}
\includegraphics[width=.33\linewidth]{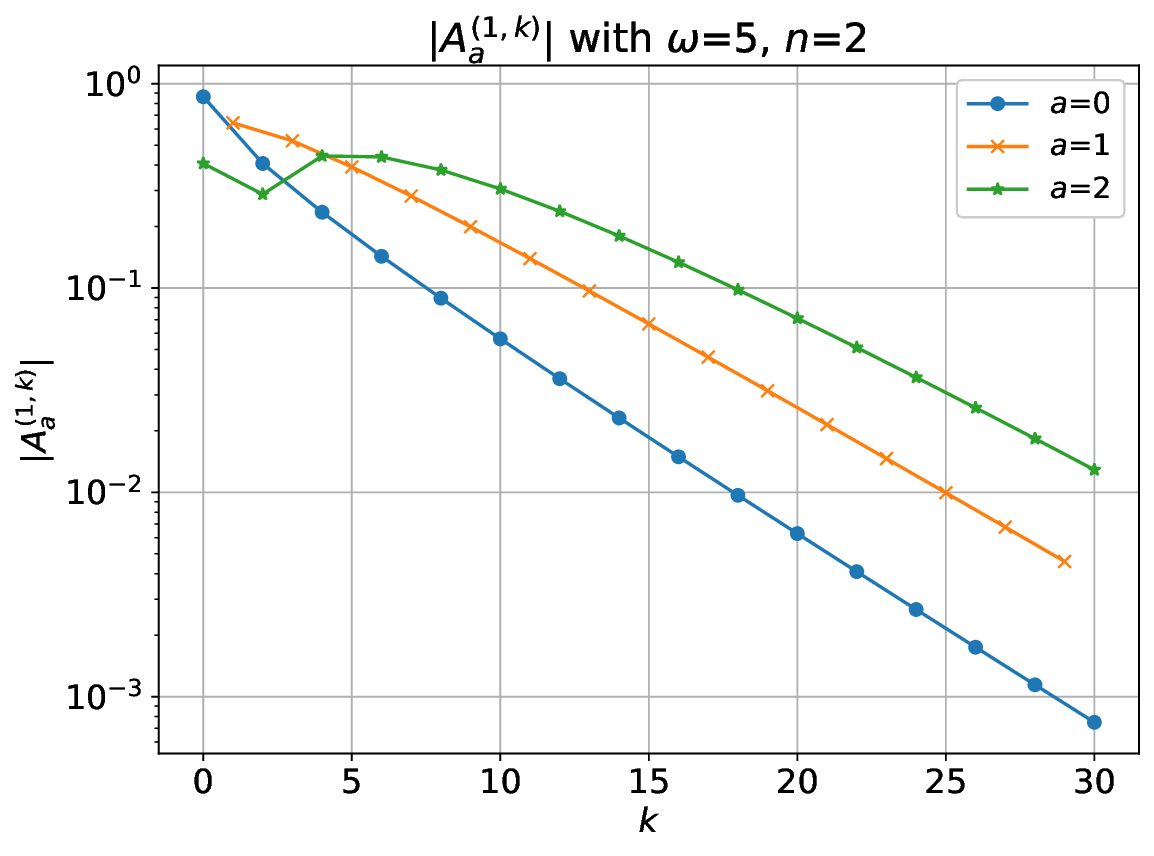}\hfill
\includegraphics[width=.33\linewidth]{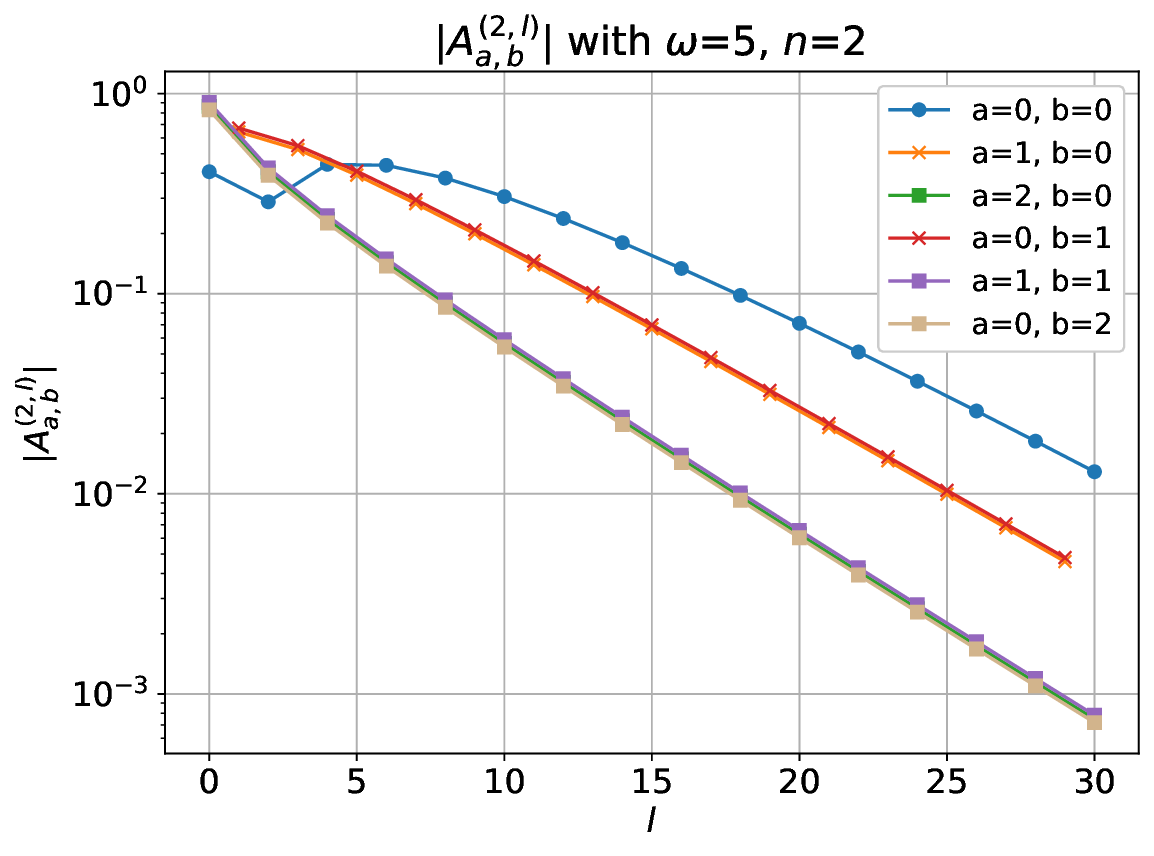}
\includegraphics[width=.33\linewidth]{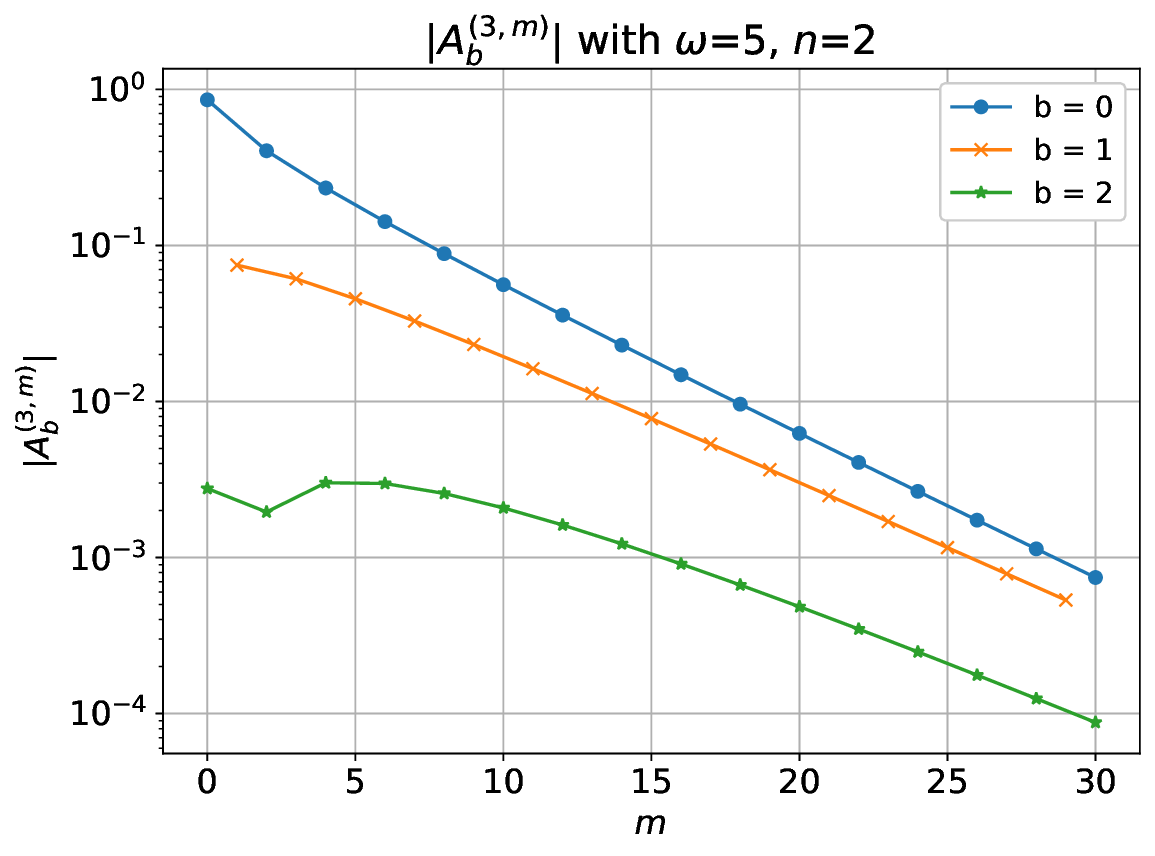}
\caption{(color online) Absolute values of the nonzero MPS matrix elements for the second excited state ($n=2$) as functions of the physical indices with $\omega=5$ and $D_{12}=D_{23}=0.25,\\ D_{13}=0.$ Some of them reach a maximum, after which they decay exponentially.}
\end{figure}

Additionally, we explicitly constructed an analytical MPS representation for certain eigenstates of a chain of three coupled quantum harmonic oscillators. 
It should be possible to generalize the results of Section 4 to the general $N$-particle case. Furthermore, it could be interesting to consider the continuous limit in the continuous MPS formalism introduced in \cite{verstraete2010continuous}.

An interesting and natural generalization of the results of this paper would be to study the hierarchical Tucker (HT) format in the context of infinite-dimensional Hilbert spaces, allowing for a more general tree-like structure of the tensor network representing vectors in tensor product spaces. 

Another interesting question is the nature of the operators acting on the auxiliary spaces in idMPS, allowing to draw parallels between MPS and operator theory on Hilbert spaces. As compact operators can be approximated by finite-rank operators, it is an interesting question under which assumptions the idMPS operators turn out to be compact. In the same context one could investigate the error introduced when approximating infinite-dimensional MPS with finite-dimensional MPS, and try to obtain analytical error estimates when truncating both in the physical and auxiliary Hilbert spaces.

Similarly as MPS have turned out useful in classically simulating certain kinds of quantum computations, idMPS could have applications in the context of continuous-variable quantum computation \cite{buck2021continuous}. Finally, considering the continuous limit of idMPS might lead to connections with continuous MPS and broaden their applicability in the study of one-dimensional quantum field theories.

\section*{Acknowledgements}
I would like to thank my supervisors Jani Lukkarinen and Paolo Muratore-Ginanneschi, as the completion of this paper would not have been possible without their invaluable support and feedback. The research has been supported by the Research Council of Finland, via the Finnish Centre of Excellence in Randomness and Structures (CoE FiRST, project No. 346306). Also, I acknowledge the support from the QDOC program of the Research Council of Finland.

\section*{Appendix}
\appendix

\section{The Schmidt Decomposition} \label{Appendix A}
The proof of the Schmidt decomposition is based on the singular value decomposition of compact operators and the fact that the tensor product of Hilbert spaces is (conjugate-)isomorphic to the space of Hilbert-Schmidt operators, which are compact.
\begin{prop} (SVD of Compact Operators) \\
		If $T \in \bddhh$ is a compact operator with rank $N\in \nn_0 \cup \{\infty \}$, then there exist
  orthonormal sets $\{ \ket{e_k} \}_{k=1}^N \subseteq \hilone$ and $\{ \ket{f_k }\}_{k=1}^N \subseteq \hiltwo$ and positive real numbers 
		$\{ \lambda_k\}_{k=1}^N$ with $\lambda_k \limk 0$ (if $N$ is infinite) such that 
		\begin{equation}\label{SVD}
			T
   =\sum_{k=1}^N \lambda_k \ket{f_k}\bra{e_k}.
		\end{equation}
		The sum, which may be infinite or finite, converges to T in operator norm. The numbers $\lambda_k$ are called the \emph{singular values} of $T$ and the expression (\ref{SVD}) the \emph{singular value decomposition} of $T$.

\begin{proof}
    See e.g. Theorem 1.6 of \cite{crane2020singular} or Theorem VI.17 of \cite{reed1981functional}.
\end{proof}
\end{prop}

\begin{prop} \label{SVD-HS}
    If $T\in\linearhshh,$ then its singular value decomposition (which exists because Hilbert-Schmidt operators are compact) converges to $T$ in both Hilbert-Schmidt and operator norms.
\begin{proof}
    The fact that the SVD converges to some operator $T$ in the Hilbert-Schmidt norm follows from orthonormality of the singular vectors and the equality $\normhs{T}^2 = \sum_{k=1}^N \lambda_k^2,$ where $\{\lambda_k\}$ are the singular values of $T$. The fact that the SVD converges to the same operator in both norms now follows from the estimate $\norm{T-S_n}_{op} \leq \normhs{T-S_n} \limn 0,$ where $S_n$ denotes the $n$th partial sum of the SVD.
\end{proof}
\end{prop}

\begin{lem}\label{HS-TP isomorphism}
Let $\hilbert_1$ and $\hilbert_2$ be separable Hilbert spaces and $\hilone^{\,*}$ the dual of $\hilone.$ The map 
 $F: \hilone^{\,*} \otimes \hiltwo \to \linearhs(\hilone, \hiltwo),$
	\begin{equation}
		F(\bra{\psi} \otimes \ket{\phi}) 
  = \ket{\phi}\bra{\psi},
	\end{equation}
    extended linearly and continuously,
	is a Hilbert space isomorphism $\hilbert^{\,*}_1 \otimes \hilbert_2 \to \linear_{HS}(\hilbert_1, \hilbert_2).$ 
 Similarly, the (conjugate-linear) map 
 $F: \hilone \otimes \hiltwo \to \linearhs(\hilone, \hiltwo),$
	\begin{equation}\label{HS-TP-isomorphism-formula}
		F(\ket{\psi} \otimes \ket{\phi}) 
  = \ket{\phi}\bra{\psi} ,
	\end{equation}
    extended linearly and continuously, isometrically identifies $\hilone\hspace{-0.2mm} \otimes \hspace{-0.2mm}\hilbert_2$ and $ \linear_{HS}(\hspace{-0.3mm}\hilbert_1\hspace{-0.1mm}, \hilbert_2\hspace{-0.2mm}).$
\end{lem}
\begin{proof}
    In §2 of \cite{Berberian2013} the Theorem is proved for conjugate-linear Hilbert-Schmidt operators and the tensor product $\tptwo,$ which is equivalent to our first claim. The second claim follows directly from the first.
\end{proof}

\begin{proof}[Proof of Proposition \ref{SD-Prop}]
By Theorem \ref{HS-TP isomorphism} we can isometrically identify any $\ketpsi \in \tptwo$ with a Hilbert-Schmidt operator $T_\psi \in \linearhshh$ via the mapping $F$ given in (\ref{HS-TP-isomorphism-formula}). As $T_\psi$ is compact, it has an SVD given by 
\begin{equation} \label{SVD in SD proof}
   T_\psi = \sum_{k=1}^{N} \lambda_k \ket{f_k} \bra{e_k},
\end{equation}
where the singular vectors are orthonormal, the singular values tend to zero and $N=\rank(T)\in\nn_0\cup\{\infty\}$.
By Proposition \ref{SVD-HS} the series (\ref{SVD in SD proof}) converges to $T_\psi$ in Hilbert-Schmidt norm.

The inverse map $F^{-1}$ is also isometric and thus continuous, and applying it on both sides of (\ref{SVD in SD proof}) yields the orthonormal series
\begin{equation}
    \ketpsi = \sum_{k=1}^{N} \lambda_k \ket{e_k} \otimes \ket{f_k},
\end{equation}
which converges in $\hilone\otimes\hiltwo.$ This is the desired Schmidt decomposition.
\end{proof}

\section{Integral in Section 4} \label{Appendix B}
In the construction of the MPS in Section 4, we encounter the integral
\begin{equation} \label{scaled-hermite-integral}
    I_{i,j} = \int_{-\infty}^\infty e^{-\frac{1}{2}( x^2 + \Tilde{\omega} x^2 )} \hermite{i}(x)\hermite{j}(\sqrt{\Tilde{\omega}}\,x)dx.
\end{equation}
This can be computed e.g. using generating functions, as demonstrated in the Mathematics Stack Exchange post \cite{hermiteintegral}, and we will use their method to obtain a formula for the integral (\ref{scaled-hermite-integral}). To this end, we can write an exponential generating function $I(s,t)$ of the integral $I_{i,j}$ as
\begin{align}
    I(s,t) &= \sum_{i=0}^\infty\sum_{j=0}^\infty I_{i,j} \frac{s^i}{i!}\frac{t^j}{j!} \nonumber \\
    &= \int_{-\infty}^\infty e^{-\frac{1}{2}( x^2 + \Tilde{\omega} x^2 )} 
    \left( \sum_{i=0}^\infty \hermite{i}(x) \frac{s^i}{i!} \right)
    \left( \sum_{j=0}^\infty \hermite{j}\left( \sqrt{\Tilde{\omega}} \, x \right) \frac{t^j}{j!} \right) dx.
\end{align}
Using the standard generating function of Hermite polynomials, we have $\sum_{i=0}^\infty \hermite{i}(x) \frac{s^i}{i!} = e^{2xs-s^2}$ and $\sum_{j=0}^\infty \hermite{j}\left( \sqrt{\Tilde{\omega}} \, x \right) \frac{t^j}{j!} = e^{2x\sqrt{\Tilde{\omega}}\,t-t^2},$ and therefore
\begin{align}
     I(s,t) 
     &= e^{-s^2-t^2} \int_{-\infty}^\infty
     e^{-\frac{1}{2}( x^2 + \Tilde{\omega} x^2 ) + 2x(s+\sqrt{\Tilde{\omega}}\,t)} dx \nonumber \\
     &= \sqrt{\frac{2\pi}{1+\Tilde{\omega}}} \exp\left( \frac{2(\sqrt{\Tilde{\omega}}t + s)^2}{1+\Tilde{\omega}} - s^2 -t^2 \right).
     \label{integral-gen-func}
\end{align}
Now the integral for any $i,j\in\nn_0$ is given by the derivative $I_{i,j} = \partial_s^i \partial_t^j I(s,t) \big{|}_{s,t=0}.$ Writing the exponential (\ref{integral-gen-func}) as a Maclaurin series and applying the multinomial formula as well as the fact $\partial_x^i x^j = \delta_{i,j}i!$ yields the expression
\begin{equation}
    I_{i,j} = \sqrt{\frac{2\pi}{1+\omegatilde}} \sum_{k=0}^\infty \sum_{p+q+r=k} 
    \frac{1}{p!q!r!}
    \frac{(-1)^p(1-\omegatilde)^{p+q}(4\sqrt{\omegatilde})^r}
    {(1+\omegatilde)^{p+q+r}} \delta_{i,2q+r}\delta_{j,2p+r}.
\end{equation}
We see in particular that if $i$ and $j$ have different parity, then $I_{i,j}=0.$

\bibliographystyle{unsrt}

\end{document}